\documentclass[12pt]{article}
\usepackage{enumerate}
\usepackage{amsfonts}
\usepackage{amsmath}
\usepackage{amssymb}
\usepackage{amsthm}
\usepackage{latexsym}
\usepackage{color}

\def\D{\mathcal{D}}
\def\H{\mathcal{H}}

\def\P{\mathcal{P}}

\def\S{\mathfrak{S}}

\def\F{\mathfrak{F}}
\def\C{\mathfrak{C}}
\def\T{\mathfrak{T}}

\def\N{\mathbb{N}}

\newcommand{\supp}{\mathrm{supp}}
\newcommand{\rank}{\mathrm{rank}}

\newcommand{\id}{\mathrm{Id}}
\newcommand{\Tr}{\mathrm{Tr}}

\newcommand{\shs}{\hspace{1pt}}
\newcounter{defin}  \newcounter{lemma}  \newcounter{theorem}
\newcounter{proposition} \newcounter{corol}  \newcounter{remark} \newcounter{example}
\newenvironment{lemma}{\par\refstepcounter{lemma}     \textbf{Lemma \thelemma.} }{\rm\par}
\newenvironment{theorem}{\par\refstepcounter{theorem}     \textbf{Theorem \thetheorem.}\ }{\rm\par}
\newenvironment{proposition}{\par\refstepcounter{proposition}     \textbf{Proposition \theproposition.}\ }{\rm\par}
\newenvironment{corollary}{\par\refstepcounter{corol}     \textbf{Corollary \thecorol.} }{\rm\par}

\newenvironment{remark}{\par\refstepcounter{remark}     \textbf{Remark \theremark.}}{\rm\par}
\newenvironment{example}{\par\refstepcounter{example}     \textbf{Example \theexample.}}{\rm\par}

\begin{document}


\title{The Alicki-Fannes-Winter technique in the quasi-classical settings: advanced version and its applications}

\author{M.E.~Shirokov\footnote{email:msh@mi.ras.ru}\\
Steklov Mathematical Institute, Moscow, Russia}
\date{}
\maketitle
\begin{abstract}
We describe an advanced version of the AFW-technique proposed in \cite{LCB,LLB} which allows us to obtain
lower semicontinuity bounds, continuity bounds and local lower bounds for characteristics of quantum systems and discrete random variables.

We consider applications of the new version of the AFW-technique to several basic characteristics of quantum systems (the von Neumann entropy, the energy-type functionals,
the quantum relative entropy, the conditional entropy and the entanglement of formation).
\end{abstract}

\tableofcontents

\section{Introduction}

In the resent articles \cite{LCB,LLB}, a special modification  of the Alicki-Fannes-Winter technique
widely used for quantitative continuity analysis of characteristics  of quantum systems (see \cite{A&F,W-CB},\cite[Section 3]{QC}) were proposed. The main advantage of this modification (called \emph{quasi-classical} in  \cite{LCB,LLB}) consists in  the fact that it provides a quite  universal way to
obtain semicontinuity bounds and local lower bounds for characteristics  of quantum systems
(resp.  discrete random variables) which, mathematically, are functions on the set of trace class positive operators with unit trace (resp.  the set of probability
distributions).

If $f(x)$ is a function on a metric space $X$ taking values in $(-\infty,+\infty]$ then a \emph{lower semicontinuity bound} for this function at a point $x_1$ with finite $f(x_1)$ is the inequality
\begin{equation}\label{SCB-g}
f(x_1)-f(x_2)\leq \mathrm{SB}_f(\varepsilon|\shs x_1)
\end{equation}
valid for any $x_2$ in $X$ such that $d(x_1,x_2)\leq\varepsilon$, where $d$ is a metric on $X$ and
it is assumed that the function $\mathrm{SB}_f(\varepsilon|\shs x_1)$ may depend on $x_1$ but cannot depend on $x_2$. The function $f(x)$ is
lower semicontinuous at a point $x_1$ if and only if (\ref{SCB-g}) holds with  the function $\mathrm{SB}_f(\varepsilon|\shs x_1)$ tending to zero as $\,\varepsilon\to0$.
In this case we say that the semicontinuity bound (\ref{SCB-g}) is \emph{faithful}.

A \emph{local lower bound} for a function $f(x)$ on a metric space $X$ at a point $x_1$ is the inequality
\begin{equation}\label{LLB-g}
f(x_2)\geq \mathrm{LB}_f(\varepsilon|\shs x_1)
\end{equation}
valid for any $x_2$ in $X$ such that $d(x_1,x_2)\leq\varepsilon$, where $d$ is a metric on $X$ and
it is assumed that the function $\mathrm{LB}_f(\varepsilon|\shs x_1)$ may depends on $x_1$ but cannot depends on $x_2$.  The function $f(x)$ is
lower semicontinuous at a point $x_1$ if and only if (\ref{LLB-g}) holds with  the function $\mathrm{LB}_f(\varepsilon|\shs  x_1)$ tending to $f(x_1)\leq+\infty$ as $\varepsilon\to0$.
In this case we say that the local lower bound (\ref{LLB-g}) is \emph{faithful}. If $f(x_1)<+\infty$ then
a faithful local lower bound (\ref{LLB-g}) can be obtained from any faithful semicontinuity bound (\ref{SCB-g}) and vice versa.
\smallskip

The  Alicki-Fannes-Winter technique in quasi-classical settings proposed in  \cite{LCB,LLB} is
used in these articles to obtain lower semicontinuity and local lower bounds for several important characteristics  of quantum systems
and discrete random variables. The aim of this article is to describe an advanced version of this technique  based on the optimisation of the main lemmas from \cite{LCB,LLB} (this optimisation was motivated by the trick used in the resent article \cite{D++}).
As a result, by applying this new version it is possible to obtain optimal or close-to-optimal semicontinuity and local lower bounds
for characteristics  of quantum systems and discrete random variables. The second part of this article is devoted to description of these new applications.

\section{Preliminaries}

Let $\mathcal{H}$ be a separable Hilbert space,
$\mathfrak{B}(\mathcal{H})$ the algebra of all bounded operators on $\mathcal{H}$ with the operator norm $\|\cdot\|$ and $\mathfrak{T}( \mathcal{H})$ the
Banach space of all trace-class
operators on $\mathcal{H}$  with the trace norm $\|\!\cdot\!\|_1$. Let
$\mathfrak{S}(\mathcal{H})$ be  the set of quantum states (positive operators
in $\mathfrak{T}(\mathcal{H})$ with unit trace) \cite{H-SCI,N&Ch,Wilde}.

Write $I_{\mathcal{H}}$ for the unit operator on a Hilbert space
$\mathcal{H}$ and $\id_{\mathcal{\H}}$ for the identity
transformation of the Banach space $\mathfrak{T}(\mathcal{H})$.

We will use the Mirsky inequality
\begin{equation}\label{Mirsky-ineq+}
  \sum_{i=1}^{+\infty}|\lambda^{\rho}_i-\lambda^{\sigma}_i|\leq \|\rho-\sigma\|_1
\end{equation}
valid for any positive operators $\rho$ and $\sigma$ in $\T(\H)$, where  $\{\lambda^{\rho}_i\}_{i=1}^{+\infty}$
and $\{\lambda^{\sigma}_i\}_{i=1}^{+\infty}$ are the sequences
of eigenvalues of $\rho$ and $\sigma$ arranged in the non-increasing order (taking the multiplicity into account) \cite{Mirsky,Mirsky-rr}.

The \emph{von Neumann entropy} of a quantum state
$\rho \in \mathfrak{S}(\H)$ is  defined by the formula
$S(\rho)=\operatorname{Tr}\eta(\rho)$, where  $\eta(x)=-x\ln x$ if $x>0$
and $\eta(0)=0$. It is a concave lower semicontinuous function on the set~$\mathfrak{S}(\H)$ taking values in~$[0,+\infty]$ \cite{H-SCI,L-2,W}.
The von Neumann entropy satisfies the inequality
\begin{equation}\label{w-k-ineq}
S(p\rho+(1-p)\sigma)\leq pS(\rho)+(1-p)S(\sigma)+h(p)
\end{equation}
valid for any states  $\rho$ and $\sigma$ in $\S(\H)$ and $p\in[0,1]$, where $\,h(p)=\eta(p)+\eta(1-p)\,$ is the binary entropy and both sides can be equal to $+\infty$ \cite{O&P,N&Ch,Wilde}.

We will use the  homogeneous extension of the von Neumann entropy to the positive cone $\T_+(\H)$ defined as
\begin{equation}\label{S-ext}
\tilde{S}(\rho)\doteq(\Tr\rho)S(\rho/\Tr\rho)=\Tr\eta(\rho)-\eta(\Tr\rho)
\end{equation}
for any nonzero operator $\rho$ in $\T_+(\H)$ and equal to $0$ at the zero operator \cite{L-2}.

By using concavity of the von Neumann entropy and inequality (\ref{w-k-ineq}) it is easy to show that
\begin{equation*}
\tilde{S}(\rho)+\tilde{S}(\sigma)\leq \tilde{S}(\rho+\sigma)\leq \tilde{S}(\rho)+\tilde{S}(\sigma)+\tilde{h}(\{\Tr\rho,\Tr\sigma\})
\end{equation*}
for any $\rho$ and $\sigma$ in $\T_+(\H)$, where  $\tilde{h}(\{\Tr\rho,\Tr\sigma\})=\eta(\Tr\rho)+\eta(\Tr\sigma)-\eta(\Tr(\rho+\sigma))$
is the homogeneous extension of the binary entropy to the positive cone in $\mathbb{R}^2$.

The \emph{quantum relative entropy} for two quantum states $\rho$ and
$\sigma$ in $\S(\H)$ is defined as
\begin{equation}\label{URE-def}
D(\rho\shs\|\shs\sigma)=\sum_i\langle
i|\,\rho\ln\rho-\rho\ln\sigma\,|i\rangle,
\end{equation}
where $\{|i\rangle\}$ is the orthonormal basis of
eigenvectors of the state $\rho$ and it is assumed that
$D(\rho\|\shs\sigma)=+\infty$ if $\,\mathrm{supp}\rho\shs$ is not
contained in $\shs\mathrm{supp}\shs\sigma$ \cite{H-SCI,Wilde}.\footnote{The support $\mathrm{supp}\rho$ of an operator $\rho$ in $\T_+(\H)$ is the closed subspace spanned by the eigenvectors of $\rho$ corresponding to its positive eigenvalues.}
If $S(\rho)<+\infty$ then
\begin{equation}\label{re-exp}
D(\rho\shs\|\shs\sigma)=\Tr\rho(-\ln\sigma)-S(\rho),
\end{equation}
where $\Tr\rho(-\ln\sigma)$ is defined according to the rule (\ref{H-fun}) described below. \smallskip

The \emph{quantum conditional entropy} (QCE) of a state $\rho$ of a finite-dimensional bipartite system $AB$ is defined as
\begin{equation}\label{ce-def}
S(A\vert B)_{\rho}=S(\rho)-S(\rho_{B}).
\end{equation}
The function $\rho\mapsto S(A\vert B)_{\rho}$ is  concave and
\begin{equation}\label{ce-ub}
\vert S(A\vert B)_{\rho}\vert\leq S(\rho_A)
\end{equation}
for any state $\rho$ in $\S(\H_{AB})$ \cite{H-SCI,Wilde}. By using concavity of the von Neumann entropy and inequality (\ref{w-k-ineq}) it is easy to show that
\begin{equation}\label{ce-LAA-2}
S(A\vert B)_{p\rho+(1-p)\sigma}\leq p S(A\vert B)_{\rho}+(1-p)S(A\vert B)_{\sigma}+h(p)
\end{equation}
for any states  $\rho$ and $\sigma$ in $\S(\H_{AB})$ and $p\in[0,1]$, where $\,h\,$ is the binary entropy.

Definition (\ref{ce-def}) remains valid for a state $\rho$ of an infinite-dimensional bipartite system $AB$
with finite marginal entropies
$S(\rho_A)$ and $S(\rho_B)$  (since the finiteness of $S(\rho_A)$ and $S(\rho_B)$ are equivalent to the finiteness of $S(\rho)$ and $S(\rho_B)$).
For a state $\rho$ with finite $S(\rho_A)$ and arbitrary $S(\rho_B)$ one can define the QCE
by the formula
\begin{equation}\label{ce-ext}
S(A\vert B)_{\rho}=S(\rho_{A})-D(\rho\shs\|\shs\rho_{A}\otimes\rho_{B})
\end{equation}
proposed and analysed by Kuznetsova in \cite{Kuz} (the finiteness of $S(\rho_{A})$ implies the finiteness of $D(\rho\shs\|\shs\rho_{A}\otimes\rho_{B})$ by the inequality (\ref{MI-UB})). The QCE extented by the above formula to the convex set $\,\{\rho\in\S(\H_{AB})\,\vert\,S(\rho_A)<+\infty\}\,$ possesses all basic properties of the QCE valid in finite dimensions \cite{Kuz}. In particular, it is concave and satisfies inequalities (\ref{ce-ub}) and (\ref{ce-LAA-2}) with possible values $+\infty$ in one or both sides.\smallskip

The \emph{quantum mutual information} of a state $\,\rho\,$ in $\S(\H_{AB})$ is defined as
\begin{equation}\label{mi-d}
I(A\!:\!B)_{\rho}=D(\rho\shs\|\shs\rho_{A}\otimes
\rho_{\shs B})=S(\rho_{A})+S(\rho_{\shs B})-S(\rho),
\end{equation}
where the second formula is valid if $\,S(\rho)\,$ is finite \cite{L-mi}.
Basic properties of the relative entropy show that $\,\rho\mapsto
I(A\!:\!B)_{\rho}\,$ is a lower semicontinuous function on the set
$\S(\H_{AB})$ taking values in $[0,+\infty]$. It is well known that
\begin{equation}\label{MI-UB}
I(A\!:\!B)_{\rho}\leq 2\min\{S(\rho_A),S(\rho_B)\}
\end{equation}
for any state $\rho\in\S(\H_{AB})$ and that factor 2 in (\ref{MI-UB}) can be omitted if $\rho$ is a separable state \cite{L-mi,Wilde}.

The quantum mutual information is not convex or concave, but it satisfies the inequalities
\begin{equation}\label{MI-LAA-1}
I(A\!:\!B)_{p\rho+(1-p)\sigma} \geq p I(A\!:\!B)_{\rho}+(1-p)I(A\!:\!B)_{\sigma}-h(p)
\end{equation}
and
\begin{equation}\label{MI-LAA-2}
I(A\!:\!B)_{p\rho+(1-p)\sigma} \leq p I(A\!:\!B)_{\rho}+(1-p)I(A\!:\!B)_{\sigma}+h(p)
\end{equation}
valid for any states $\rho$ and $\sigma$ in $\S(\H_{AB})$ and  any $p\in(0,1)$  with possible values $+\infty$ in both sides \cite[Section 4.2]{QC}.
\smallskip

Let $H$ be a positive (semi-definite)  operator on a Hilbert space $\mathcal{H}$ (we will always assume that positive operators are self-adjoint). Write  $\mathcal{D}(H)$ for the domain of $H$. For any positive operator $\rho\in\T(\H)$ we will define the quantity $\Tr H\rho$ by the rule
\begin{equation}\label{H-fun}
\Tr H\rho=
\left\{\begin{array}{ll}
        \sup_n \Tr P_n H\rho\;\; &\textrm{if}\;\;  \supp\rho\subseteq {\rm cl}(\mathcal{D}(H))\\
        +\infty\;\;&\textrm{otherwise}
        \end{array}\right.
\end{equation}
where $P_n$ is the spectral projector of $H$ corresponding to the interval $[0,n]$ and ${\rm cl}(\mathcal{D}(H))$ is the closure of $\mathcal{D}(H)$. If
$H$ is the Hamiltonian (energy observable) of a quantum system described by the space $\H$ then
$\Tr H\rho$ is the mean energy of a state $\rho$.

For any positive operator $H$ the set
$$
\C_{H,E}=\left\{\rho\in\S(\H)\,|\,\Tr H\rho\leq E\right\}
$$
is convex and closed (since the function $\rho\mapsto\Tr H\rho$ is affine and lower semicontinuous). It is nonempty if $E> E_0$, where $E_0$ is the infimum of the spectrum of $H$.

The von Neumann entropy is continuous on the set $\C_{H,E}$ for any $E> E_0$ if and only if the operator $H$ satisfies  the \emph{Gibbs condition}
\begin{equation}\label{H-cond}
  \Tr\, e^{-\beta H}<+\infty\quad\textrm{for all}\;\,\beta>0
\end{equation}
and the supremum of the entropy on this set is attained at the \emph{Gibbs state}
\begin{equation*}
\gamma_H(E)\doteq e^{-\beta_H(E) H}/Z_H(E),
\end{equation*}
where the parameter $\beta_H(E)$ is determined by the equation $\Tr H e^{-\beta H}=E\Tr e^{-\beta H}$,
\begin{equation}\label{Z}
  Z_H(E) \doteq \Tr e^{-\beta_H(E)H}
\end{equation}
and it is assumed that $\gamma_H(E)|\varphi\rangle=0$ for any $\varphi\in \D(H)^{\perp}$ \cite{W}. Condition (\ref{H-cond}) can be valid only if $H$ is an unbounded operator having  discrete spectrum of finite multiplicity. It means, in Dirac's notation, that
\begin{equation}\label{H-form}
H=\sum_{k=1}^{+\infty} h_k |\tau_k\rangle\langle\tau_k|,
\end{equation}
where
$\mathcal{T}\doteq\left\{\tau_k\right\}_{k=1}^{+\infty}$ is the orthonormal
system of eigenvectors of $H$ corresponding to the \emph{nondecreasing} unbounded sequence $\left\{h_k\right\}_{k=1}^{+\infty}$ of its eigenvalues
and \emph{it is assumed that the domain $\D(H)$ of $H$ lies within the closure $\H_\mathcal{T}$ of the linear span of $\mathcal{T}$}. In this case
\begin{equation*}
\Tr H \rho=\sum_i \lambda_i\|\sqrt{H}\varphi_i\|^2
\end{equation*}
for any operator $\rho$ in $\T_+(\H)$ with the spectral decomposition $\rho=\sum_i \lambda_i|\varphi_i\rangle\langle\varphi_i|$ provided that
all the vectors $\varphi_i$ lie in $\D(\sqrt{H})=\{ \varphi\in\H_\mathcal{T}\,| \sum_{k=1}^{+\infty} h_k |\langle\tau_k|\varphi\rangle|^2<+\infty\}$. If at least one eigenvector of $\rho$ corresponding to a nonzero eigenvalue does not belong to the set $\D(\sqrt{H})$
then $\Tr H \rho=+\infty$.


For any positive operator $H$ of the form (\ref{H-form}) we will use the function
\begin{equation}\label{F-def}
F_{H}(E)\doteq\sup_{\rho\in\C_{H,E}}S(\rho)
\end{equation}
on $[h_1,+\infty)$ which is finite and concave provided that $\Tr e^{-\beta H}<\infty$  (properties of this function is described in Proposition 1 in \cite{EC}).
If the operator $H$ satisfies the Gibbs condition (\ref{H-cond}) then
$$
F_{H}(E)=S(\gamma_H(E))=\beta_H(E)E+\ln Z_H(E).
$$
It is easy to see that $F_{H}(h_1)=\ln m(E_0)$, where $m(E_0)$ is the multiplicity of $h_1$. By Proposition 1 in \cite{EC} the Gibbs condition (\ref{H-cond}) is equivalent to the following asymptotic property
\begin{equation}\label{H-cond-a}
  F_{H}(E)=o\shs(E)\quad\textrm{as}\quad E\rightarrow+\infty.
\end{equation}

For example, if $\,\hat{N}\doteq a^\dag a\,$ is the number operator of a quantum oscillator then $F_{\hat{N}}(E)=g(E)$, where
\begin{equation}\label{g-def}
 g(x)=(x+1)\ln(x+1)-x\ln x,\;\, x>0,\qquad g(0)=0.
\end{equation}

We will often assume that
\begin{equation}\label{star}
  h_1=\inf\limits_{\|\varphi\|=1}\langle\varphi\vert H\vert\varphi\rangle=0.
\end{equation}
In this case the function $F_H$ satisfies the conditions of the following lemma (cf.\cite[Corollary 12]{W-CB}).\smallskip

\begin{lemma}\label{W-L} \emph{If $f$ is a nonnegative concave function on $\mathbb{R}_+$ then
\begin{equation*}
  xf(z/x)\leq yf(z/y)\quad \textit{ for all }\;y>x>0\;\textit{ and }\;z\geq0.
\end{equation*}}
\end{lemma}

For a nonnegative function $f$ on $[0,1]$ we denote  by $f^{\uparrow}$ its non-decreasing
envelope, i.e.
\begin{equation*}
  f^{\uparrow}(x)\doteq \sup_{t\in[0,x]} f(t),\quad x\in[0,1].
\end{equation*}
In particular, we will often use  the \emph{nondecreasing concave} continuous function
\begin{equation}\label{h+}
h^{\uparrow}(x)=\left\{\begin{array}{ll}
        h(x) &\textrm{if}\;\;  x\in[0,\frac{1}{2}]\\
        \ln2\quad& \textrm{if}\;\;  x\in(\frac{1}{2},1].
        \end{array}\right.
\end{equation}
Dealing with a multivariate expression $f(x,y,..)$ in which $x\in[0,1]$ we will assume that
\begin{equation}\label{men-def+}
  \{f(x,y,..)\}^{\uparrow}_x\doteq \sup_{t\in[0,x]} f(t,y,..).
\end{equation}

\section{Advanced version of the Alicki-Fannes-Winter technique in the quasi-classical settings}\label{sec3}

\subsection{New basic lemmas}

In this subsection we describe two general results concerning properties of a function $f$ on a convex subset $\S_0$ of $\S(\H)$ taking values in $(-\infty,+\infty]$ and satisfying the inequalities
\begin{equation}\label{LAA-1}
  f(p\rho+(1-p)\sigma)\geq pf(\rho)+(1-p)f(\sigma)-a_f(p)
\end{equation}
and
\begin{equation}\label{LAA-2}
  f(p\rho+(1-p)\sigma)\leq pf(\rho)+(1-p)f(\sigma)+b_f(p)
\end{equation}
(with possible value $+\infty$ in both sides) for all states $\rho$ and $\sigma$ in $\S_0$ and any $p\in[0,1]$, where $a_f$ and $b_f$ are continuous  functions on $[0,1]$ such that $a_f(0)=b_f(0)=0$.
These  inequalities can be treated, respectively, as weakened forms of concavity and convexity. We will call functions
satisfying both inequalities (\ref{LAA-1}) and (\ref{LAA-2}) \emph{locally almost affine} (briefly, \emph{LAA functions}), since for any such function $f$ the quantity
$\,\vert f(p\rho+(1-p)\sigma)-p f(\rho)-(1-p)f(\sigma)\vert \,$ tends to zero as $\,p\rightarrow 0^+$ uniformly on $\,\S_0\times\S_0$.
For technical simplicity we will also assume that
\begin{equation}\label{a-b-assump}
 \textrm{ the functions}\;\; a_f\;\; \textrm{and}\;\; b_f\;\; \textrm{are concave and non-decreasing on}\;\; \textstyle[0,\frac{1}{2}].
\end{equation}

For any function $f$ on a convex subset $\S_0$ of $\,\S(\H)$ we will use its  homogeneous extension $\tilde{f}$ to the cone $\widetilde{\S}_0$ generated by $\S_0$. It is defined as
\begin{equation}\label{G-ext}
\tilde{f}(\rho)\doteq(\Tr\rho)f(\rho/\Tr\rho),\quad \tilde{f}(0)=0.
\end{equation}

\smallskip

Let $\{X,\F\}$ be a measurable space and $\tilde{\omega}(x)$ a $\F$-measurable $\S(\H)$-valued function on $X$. Denote by $\P(X)$ the set of all probability measures on $X$ (more precisely, on $\{X,\F\}$). We will assume that the function $\tilde{\omega}(x)$ is integrable (in the Pettis sense \cite{P-int}) w.r.t. any measure in $\P(X)$. Consider  the set of states
\begin{equation}\label{q-set}
\mathfrak{Q}_{X,\F,\tilde{\omega}}\doteq\left\{\rho\in\S(\H)\,\left\vert\,\exists\mu_{\rho}\in\P(X):\rho=\int_X\tilde{\omega}(x)\mu_{\rho}(dx)\;\right.\right\}.
\end{equation}
We will call any measure $\mu_{\rho}$ in $\P(X)$ such that $\rho=\int_X\tilde{\omega}(x)\mu_{\rho}(dx)$ a \emph{representing measure} for a state $\rho$ in $\mathfrak{Q}_{X,\F,\tilde{\omega}}$. The state $\int_X\tilde{\omega}(x)\nu(dx)$ for any $\nu\in\P(X)$ will be denoted by $\Omega(\nu)$ for brevity.

We will use the total variation distance between probability measures $\mu$ and $\nu$ in $\P(X)$ defined as
\begin{equation*}
\mathrm{TV}(\mu,\nu)=\sup_{A\in\F}\vert\mu(A)-\nu(A)\vert.
\end{equation*}

Concrete  examples of sets $\mathfrak{Q}_{X,\F,\tilde{\omega}}$ can be found in \cite[Section 3]{LCB}.\smallskip

The following lemma is an advanced version of Lemma 1 in \cite{LCB}. It gives semicontinuity bounds for  LAA functions on a set of quantum states having form (\ref{q-set}).\smallskip

\begin{lemma}\label{g-ob} \emph{Let $\mathfrak{Q}_{X,\F,\tilde{\omega}}$ be the set defined in (\ref{q-set}) and $\S_0$ a convex subset of $\S(\H)$ such that
}\begin{equation}\label{S-prop}
  \varrho\in\S_0\cap\mathfrak{Q}_{X,\F,\tilde{\omega}}\quad \Rightarrow \quad\{\varsigma\in\mathfrak{Q}_{X,\F,\tilde{\omega}}\,\vert\,\exists c>0:c\varsigma\leq \varrho\}\subseteq\S_0.
\end{equation}

\emph{Let $f$ be a function  on the set $\,\S_0$ taking values in $(-\infty,+\infty]$ and satisfying inequalities (\ref{LAA-1}) and (\ref{LAA-2}). Let $\rho$ and $\sigma$ be states in $\,\mathfrak{Q}_{X,\F,\shs\tilde{\omega}}\cap\S_0\,$ with representing measures $\mu_{\rho}$ and $\mu_{\sigma}$ correspondingly. If $\,f(\rho)<+\infty\,$ then
\begin{equation}\label{AFW-1+}
f(\rho)-f(\sigma)\leq \varepsilon C_f(\rho,\sigma\vert\shs\varepsilon)+a_f(\varepsilon)+b_f(\varepsilon),
\end{equation}
where $\varepsilon=\mathrm{TV}(\mu_{\rho},\mu_{\sigma})$,
\begin{equation}\label{C-f}
C_f(\rho,\sigma\vert\shs\varepsilon)\doteq\sup\left\{f(\Omega(\mu))-f(\Omega(\nu))\left\vert\, \mu,\nu\in\P(X),\;\varepsilon\mu\leq\mu_{\rho},\;\varepsilon\nu\leq\mu_{\sigma},\, \mu\,\bot\,\nu\right.\right\}
\end{equation}
and the left hand side of (\ref{AFW-1+}) may be equal to $-\infty$.}\smallskip

\emph{If the function $f$ is nonnegative then inequality (\ref{AFW-1+}) holds with $C_f(\rho,\sigma\vert\shs\varepsilon)$ replaced by}
\begin{equation}\label{C-f+}
C^+_f(\rho\shs\vert\shs\varepsilon)\doteq\sup\left\{f(\Omega(\mu))\,\left\vert\, \mu\in\P(X),\;\varepsilon\mu\leq\mu_{\rho}\right.\right\}.
\end{equation}\end{lemma}

\begin{proof} We may assume that $f(\sigma)<+\infty$ and $\varepsilon\in(0,1)$, since otherwise (\ref{AFW-1+}) holds trivially. By the condition we have
$$
2\mathrm{TV}(\mu_{\rho},\mu_{\sigma})=[\mu_{\rho}-\mu_{\sigma}]_+(X)+[\mu_{\rho}-\mu_{\sigma}]_-(X)=2\varepsilon,
$$
where $[\mu_{\rho}-\mu_{\sigma}]_+$ and $[\mu_{\rho}-\mu_{\sigma}]_-$ are the positive and negative parts
of the measure $\mu_{\rho}-\mu_{\sigma}$ (in the sense of Jordan decomposition theorem \cite{Bil}). Since $\,\mu_{\rho}(X)=\mu_{\sigma}(X)=1$, it follows from the above equality that
$\,[\mu_{\rho}-\mu_{\sigma}]_+(X)=[\mu_{\rho}-\mu_{\sigma}]_-(X)=\varepsilon$. Hence, $\nu_\pm\doteq \varepsilon^{-1}[\mu_{\rho}-\mu_{\sigma}]_\pm\in\P(X)$.
Moreover, it is easy to show, by using the definition of $\,[\mu_{\rho}-\mu_{\sigma}]_\pm$ via the Hahn decomposition of $X$, that
\begin{equation}\label{m-r}
\varepsilon\nu_+=[\mu_{\rho}-\mu_{\sigma}]_+\leq \mu_{\rho}\quad \textrm{and} \quad \varepsilon\nu_-=[\mu_{\rho}-\mu_{\sigma}]_-\leq \mu_{\sigma}.
\end{equation}

Modifying the idea used in \cite{D++} consider the states
\begin{equation}\label{tau-s}
\tau_+=\Omega(\nu_+),\quad\tau_-=\Omega(\nu_-)\quad \textrm{and} \quad \omega_*=\Omega(\mu_*),
\end{equation}
where
$$
\mu_*=\frac{\mu_{\rho}-\varepsilon\nu_+}{1-\varepsilon}=\frac{\mu_{\sigma}-\varepsilon\nu_-}{1-\varepsilon}
$$
is a measure in $\P(X)$. Since the inequalities in (\ref{m-r}) imply that $\varepsilon\tau_+\leq\rho$, $\varepsilon\tau_-\leq\sigma$ and $(1-\varepsilon)\omega_*\leq\rho$, these states
belong to the set $\mathfrak{Q}_{X,\F,\tilde{\omega}}\cap\S_0$ due to condition (\ref{S-prop}).
We have
\begin{equation}\label{omega-star}
\rho=\varepsilon\tau_++(1-\varepsilon)\omega_*\quad \textrm{and} \quad\sigma=\varepsilon\tau_-+(1-\varepsilon)\omega_*.
\end{equation}
By applying inequalities (\ref{LAA-1}) and (\ref{LAA-2}) to the decompositions  in (\ref{omega-star}) we  obtain
$$
f(\rho)\leq\varepsilon f(\tau_+)+(1-\varepsilon)f(\omega_*)+b_f(\varepsilon)\quad \textrm{and} \quad f(\sigma)\geq\varepsilon f(\tau_-)+(1-\varepsilon)f(\omega_*)-a_f(\varepsilon).
$$
The last inequality implies the finiteness of $f(\tau_-)$ and $f(\omega_*)$ by the assumed finiteness of $f(\sigma)$. So,
these inequalities show that
\begin{equation}\label{AFW-1}
f(\rho)-f(\sigma)\leq\varepsilon(f(\tau_+)-f(\tau_-))+a_f(\varepsilon)+b_f(\varepsilon).
\end{equation}
This implies inequality (\ref{AFW-1+}) due to the inequalities in (\ref{m-r})
and because $\nu_+\,\bot\,\nu_-$ by the construction.\smallskip

The last claim of the lemma is obvious.
\end{proof}


\begin{remark}\label{g-ob-r} In general, the condition $\mathrm{TV}(\mu_{\rho},\mu_{\sigma})=\varepsilon$ in
Lemma \ref{g-ob} can not be replaced by the condition $\mathrm{TV}(\mu_{\rho},\mu_{\sigma})\leq\varepsilon$, since
the functions $\varepsilon \mapsto C_f(\rho,\sigma\vert\shs\varepsilon)$ and $\varepsilon \mapsto C^+_f(\rho\shs\vert\shs\varepsilon)$ may be decreasing and, hence, special arguments are required to show that the r.h.s. of (\ref{AFW-1+}) is a nondecreasing function of $\varepsilon$.
\end{remark}\smallskip

\begin{remark}\label{S-prop-r} Condition (\ref{S-prop}) in Lemma \ref{g-ob} can be replaced by the condition
\begin{equation}\label{S-prop+}
  \rho,\sigma\in\S_0\cap\mathfrak{Q}_{X,\F,\tilde{\omega}}\quad \Rightarrow \quad \tau_+,\tau_-,\omega_*\in\S_0,
\end{equation}
where $\tau_+$, $\tau_-$ and $\omega_*$ are the states defined in (\ref{tau-s}) via the representing measures $\nu_+$, $\nu_-$ and $\mu_{*}$. In this case
one should correct the definitions of $C_f(\rho,\sigma\vert\shs\varepsilon)$ and $C^+_f(\rho\shs\vert\shs\varepsilon)$
by adding the requirements $\Omega(\mu)\in\S_0$ and $\Omega(\nu)\in\S_0$ in (\ref{C-f}) and the  requirement $\Omega(\mu)\in\S_0$ in (\ref{C-f+}).
\end{remark}\smallskip\medskip

\begin{remark}\label{S-prop-r+}
In some cases it is reasonable to apply inequality (\ref{AFW-1}) (instead of (\ref{AFW-1+})) to obtain more accurate estimates (see the proofs of
Theorems  \ref{main-1} and \ref{main-2} in Section 3.2).
\end{remark}\medskip


For any state $\rho$ with the spectral representation $\,\rho=\sum_i\lambda_i^{\rho}|\varphi_i\rangle\langle\varphi_i|\,$
and arbitrary $\varepsilon>0\,$ introduce the trace class positive operators
\begin{equation}\label{2-op}
\! \rho\wedge\varepsilon I_{\H}\doteq \sum_i \min\{\lambda_i^{\rho},\varepsilon\}|\varphi_i\rangle\langle\varphi_i|,\qquad
[\rho-\varepsilon I_{\H}]_+\doteq \sum_i \max\{\lambda_i^{\rho}-\varepsilon,0\}|\varphi_i\rangle\langle\varphi_i|.
\end{equation}
The operator $[\rho-\varepsilon I_{\H}]_+$ is the positive  part of the Hermitian operator $\rho-\varepsilon I_{\H}$,
the operator $\rho\wedge\varepsilon I_{\H}$ can be defined as  $\rho-[\rho-\varepsilon I_{\H}]_+$. It is clear that $\,\rho\wedge\varepsilon I_{\H}\to0\,$ in the trace norm as $\,\varepsilon\to0^+$.
\smallskip

The following lemma is an advanced version of Lemma 1 in \cite{LLB}. It allows us to obtain semicontinuity bounds and local lower bounds for  nonnegative LAA functions  valid for commuting states, i.e. such states $\rho$ and $\sigma$ in $\,\S(\H)$ that
$[\rho,\sigma]\doteq \rho\,\sigma-\sigma\rho=0$.

\smallskip

\begin{lemma}\label{b-lemma} \emph{Let $\,\S_0$ be a convex subset of $\,\S(\H)$  such that
\begin{equation}\label{S-prop+-+}
 \{\varrho\in\S_0\}\wedge\{\exists c>0:c\varsigma\leq\rho \}\wedge\{[\varrho,\varsigma]=0\}\;\Rightarrow\;\varsigma\in\S_0.
\end{equation}
Let $f$ be a function on $\S_0$  taking values in $[0,+\infty]$ and satisfying inequalities (\ref{LAA-1}) and (\ref{LAA-2}) on $\S_0$ with possible values $+\infty$ in both sides.}\medskip\pagebreak

\noindent A) \emph{If $\,f(\rho)<+\infty\,$ then
\begin{equation}\label{g-ob-c+0}
f(\rho)-f(\sigma)\leq\tilde{f}(\rho\wedge\varepsilon I_{\H})+2a^{\uparrow}_f(\varepsilon)+b^{\uparrow}_f(\varepsilon)
\end{equation}
for any state $\sigma$ in $\,\S_0$ such that $\,[\rho,\sigma]=0\,$ and $\,\frac{1}{2}\|\rho-\sigma\|_1\leq\varepsilon\leq1$,
where $\rho\wedge\varepsilon I_{\H}$ is the operator defined in (\ref{2-op}), $\tilde{f}$ is the homogeneous extension of $\shs f$ defined in (\ref{G-ext}),
\begin{equation*}
 a^{\uparrow}_f(\varepsilon)=\left\{\begin{array}{l}
        a_f(\varepsilon)\;\; \shs\textrm{if}\;\;  \varepsilon\in\shs[0,\frac{1}{2}]\\
        a_f(\frac{1}{2})\;\; \textrm{if}\;\;  \varepsilon\in(\frac{1}{2},1]
        \end{array}\right.\quad \textit{and}\;\;\quad b^{\uparrow}_f(\varepsilon)=\left\{\begin{array}{l}
        b_f(\varepsilon)\;\; \shs\textrm{if}\;\;  \varepsilon\in\shs[0,\frac{1}{2}]\\
        b_f(\frac{1}{2})\;\; \textrm{if}\;\;  \varepsilon\in(\frac{1}{2},1].
        \end{array}\right.
\end{equation*}
}

\noindent B) \emph{ If either $\,f(\rho)<+\infty\,$ or $\,\tilde{f}([\rho-\varepsilon I_{\H}]_+)<+\infty$ for all $\varepsilon\in(0,1]$,
where $[\rho-\varepsilon I_{\H}]_+$ is the operator defined in (\ref{2-op}), then
\begin{equation}\label{g-ob-c+}
f(\sigma)\geq \tilde{f}([\rho-\varepsilon I_{\H}]_+)-2a^{\uparrow}_f(\varepsilon)-b^{\uparrow}_f(\varepsilon)-a_f(1-r_{\varepsilon}),\quad r_{\varepsilon}=\Tr[\rho-\varepsilon I_{\H}]_+,
\end{equation}
for any state $\sigma$ in $\,\S_0$ such that $\,[\rho,\sigma]=0\,$ and $\,\frac{1}{2}\|\rho-\sigma\|_1\leq\varepsilon\leq1$.}
\end{lemma}\smallskip

\begin{remark}\label{ttt}
 All the terms in the r.h.s. of (\ref{g-ob-c+0}) and (\ref{g-ob-c+}) are well defined and finite. Indeed,
the operators $\rho\wedge\varepsilon I_{\H}$ and $[\rho-\varepsilon I_{\H}]_+$ belong to the cone $\widetilde{\S}_0$ due to condition (\ref{S-prop+-+}).
If $\,f(\rho)<+\infty\,$ then inequality (\ref{LAA-1}) implies that $\tilde{f}(\rho\wedge\varepsilon I_{\H})<+\infty$ and $\tilde{f}([\rho-\varepsilon I_{\H}]_+)<+\infty$, since  $\rho\wedge\varepsilon I_{\H}+[\rho-\varepsilon I_{\H}]_+=\rho$.\smallskip
\end{remark}

\emph{Proof.} Everywhere further we will assume that $f(\sigma)<+\infty$, since otherwise (\ref{g-ob-c+0}) and (\ref{g-ob-c+}) hold trivially.
Assume that $\rho$ and $\sigma$ are commuting states in $\,\S_0$ such that $\,\frac{1}{2}\|\rho-\sigma\|_1=\epsilon\leq\varepsilon\leq1$.  Let $\{\varphi_k\}_{k=1}^{+\infty}$ be an orthonormal
basis in $\H$ such that
\begin{equation*}
\rho=\sum_{k=1}^{+\infty} \lambda^{\rho}_k |\varphi_k\rangle\langle \varphi_k|\quad \textrm{and} \quad \sigma=\sum_{k=1}^{+\infty} \lambda^{\sigma}_k |\varphi_k\rangle\langle \varphi_k|,
\end{equation*}
where $\{\lambda^{\rho}_k\}$ is the sequence of eigenvalues of $\rho$ \emph{arranged in the non-increasing order} and
$\{\lambda^{\sigma}_k\}$ is the corresponding sequence of eigenvalues of $\sigma$ (not nondecreasing, in general).\smallskip

A) Motivating by the trick  used in \cite{D++} introduce the states
\begin{equation*}
\tau_+=\epsilon^{-1}\sum_{k=1}^{+\infty} [\lambda^{\rho}_k-\lambda^{\sigma}_k]_+ |\varphi_k\rangle\langle \varphi_k|,\qquad \tau_-=\epsilon^{-1}\sum_{k=1}^{+\infty} [\lambda^{\rho}_k-\lambda^{\sigma}_k]_- |\varphi_k\rangle\langle \varphi_k|,
\end{equation*}
and
$$
\omega_*=\frac{\rho-\epsilon\tau_+}{1-\epsilon}=\frac{\sigma-\epsilon\tau_-}{1-\epsilon},
$$
where $\,[\lambda^{\rho}_k-\lambda^{\sigma}_k]_+=\max\{\lambda^{\rho}_k-\lambda^{\sigma}_k,0\}$ and $\,[\lambda^{\rho}_k-\lambda^{\sigma}_k]_-=\max\{\lambda^{\sigma}_k-\lambda^{\rho}_k,0\}$.
The states $\tau_+$, $\tau_-$ and  $ \omega_*$ belong to the set $\S_0$ due to the condition (\ref{S-prop+-+}). Then we have
\begin{equation}\label{omega-star+}
\rho=\epsilon\tau_++(1-\epsilon)\omega_*\qquad \textrm{and} \qquad\sigma=\epsilon\tau_-+(1-\epsilon)\omega_*.
\end{equation}
By applying inequalities (\ref{LAA-1}) and (\ref{LAA-2}) to the decompositions  in (\ref{omega-star+}) we  obtain
$$
f(\rho)\leq\epsilon f(\tau_+)+(1-\epsilon)f(\omega_*)+b_f(\epsilon)\qquad \textrm{and} \qquad f(\sigma)\geq\epsilon f(\tau_-)+(1-\epsilon)f(\omega_*)-a_f(\epsilon).
$$
The last inequality implies the finiteness of $f(\tau_-)$ and $f(\omega_*)$ by the assumed finiteness of $f(\sigma)$. So,
these inequalities show that
\begin{equation}\label{imp-ineq}
f(\rho)-f(\sigma)\leq\epsilon(f(\tau_+)-f(\tau_-))+a_f(\epsilon)+b_f(\epsilon)\leq \epsilon f(\tau_+)+a^{\uparrow}_f(\varepsilon)+b^{\uparrow}_f(\varepsilon),
\end{equation}
where the last inequality follows from the nonnegativity of $f$ and the monotonicity of the functions $a^{\uparrow}_f$  and $b^{\uparrow}_f$  on $[0,1]$.
Since $\epsilon\tau_+\leq\rho$, $\tau_+\leq I_{\H}$ and $[\tau_+,\rho]=0$, we have
\begin{equation*}
 \epsilon\tau_+\leq \rho\wedge\epsilon I_{\H}\leq \rho\wedge\varepsilon I_{\H}.
\end{equation*}
So,
by using inequality (\ref{LAA-1}) and taking the nonnegativity of $f$ into account we obtain
\begin{equation}\label{imp-ineq+}
\epsilon f(\tau_+)\leq \bar{r}_{\varepsilon}f(\bar{r}^{-1}_{\varepsilon}\rho\wedge\varepsilon I_{\H})+\bar{r}_{\varepsilon}a_f(\epsilon/\bar{r}_{\varepsilon})\leq \bar{r}_{\varepsilon}f(\bar{r}^{-1}_{\varepsilon}\rho\wedge\varepsilon I_{\H})+a^{\uparrow}_f(\varepsilon),
\end{equation}
where $\bar{r}_{\varepsilon}=1-r_{\varepsilon}=\Tr(\rho\wedge\varepsilon I_{\H})$ and the last inequality follows from Lemma \ref{W-L} and the assumption (\ref{a-b-assump}).
This inequality and (\ref{imp-ineq}) imply (\ref{g-ob-c+0}).\smallskip

B) This claim  is derived from claim A by the same arguments that are used to derive claim B of Lemma 1 in \cite{LLB} from part A of that lemma. $\Box$
\smallskip

\subsection{General results}

Among characteristics of a $n$-partite
finite-dimensional quantum system $A_{1}...A_{n}$ there are many that satisfy inequalities (\ref{LAA-1}) and (\ref{LAA-2})
with the functions $a_f$ and $b_f$ proportional to the binary entropy (defined after (\ref{w-k-ineq})) as well as  the  double  inequality
\begin{equation}\label{Cm}
-c^-_f C_m(\rho)\leq f(\rho)\leq c^+_f C_m(\rho),\quad C_m(\rho)=\sum_{k=1}^m S(\rho_{A_k}),\;\; m\leq n,
\end{equation}
for any state $\rho$ in $\S(\H_{A_{1}..A_{n}})$, where $c^-_f$ and $c^+_f$ are nonnegative numbers.

Following \cite{QC,LCB} introduce the class $L_n^m(C,D\vert\,\S_0)$, $m\leq n$, of functions on a convex subset $\S_0$ of $\S(\H_{A_{1}..A_{n}})$
satisfying inequalities  (\ref{LAA-1}) and (\ref{LAA-2})
with $a_f(p)=d_f^-h(p)$ and $b_f(p)=d_f^+h(p)$ and inequality (\ref{Cm}) for any states in $\S_0$ with the parameters $c^{\pm}_f$ and  $d^{\pm}_f$ such that $c^-_f+c^+_f=C$ and $d^-_f+d^+_f=D$.

If $A_1$,...,$A_n$ are arbitrary infinite-dimensional quantum systems then we define the classes $L_n^m(C,D\vert\,\S_0)$, $m\leq n$, by the same rule assuming that all functions in $L_n^m(C,D\vert\,\S_0)$ take values in $(-\infty,+\infty]$ on $\S_0$  and that all the
inequalities in (\ref{LAA-1}), (\ref{LAA-2}) and (\ref{Cm}) hold with possible infinite values in one or all sides.

We also introduce the class $\widehat{L}^{m}_n(C,D\vert\,\S_0)$ obtained by adding to the class $L^{m}_n(C,D\vert\,\S_0)$ all functions of the form
$$
f(\rho)=\inf_{\lambda}f_{\lambda}(\rho)\quad \textrm{and} \quad f(\rho)=\sup_{\lambda}f_{\lambda}(\rho),
$$
where $\{f_{\lambda}\}$ is some family of functions in $L^{m}_n(C,D\vert\,\S_0)$.

If $\S_0=\S_m(\H_{A_{1}..A_{n}})\doteq\left\{\rho\in\mathfrak{S}(\mathcal{H}_{A_1..A_n})\,\vert\,S(\rho_{A_1}),..., S(\rho_{A_m})<+\infty\shs\right\}$,
then we will denote the classes $L^{m}_n(C,D\vert\,\S_0)$ and $\widehat{L}^{m}_n(C,D\vert\,\S_0)$
by $L^{m}_n(C,D)$ and $\widehat{L}^{m}_n(C,D)$ for brevity.\footnote{The necessity to consider the case $\,\S_0\neq\S_m(\H_{A_{1}..A_{n}})$ will be shown in Section 4.}

For example, the von Neumann entropy belongs to the class $L_1^1(1,1\vert\,\S(\H))$, while the (extended) quantum conditional entropy $S(A_1\vert A_2)$ (defined in (\ref{ce-ext})) lies in $L_2^1(2,1)\doteq L_2^1(2,1\vert\,\S_1(\H_{A_1A_2}))$. This follows from the concavity of these characteristics and inequalities (\ref{w-k-ineq}), (\ref{ce-ub}) and (\ref{ce-LAA-2}). The nonnegativity of $S(A_1\vert A_2)$
on the convex set $\S^1_{\rm sep}$ of separable states in $\S_1(\H_{A_1A_2})$ implies that this function also belongs to the class $L_2^1(1,1\vert\,\S^1_{\rm sep})$. Other examples of characteristics of quantum systems belonging the classes $L^{m}_n(C,D)$ and $\widehat{L}^{m}_n(C,D)$
can be found below and in \cite[Section 4]{QC}.

Now we apply Lemma \ref{g-ob} in Section 3.1 to obtain  continuity and semicontinuity bounds for functions from the classes
$\widehat{L}_n^{m}(C,D\vert\,\S_0)$ under the rank constraint on the marginal states corresponding to the subsystems $A_1,...,A_m$.
The following theorem is an advanced version of Theorem 1 in \cite{LCB}.\smallskip

\begin{theorem}\label{main-1}\emph{ Let  $\,\mathfrak{Q}_{X,\F,\tilde{\omega}}$ be the set of states in $\,\S(\H_{A_1...A_n})$ defined in (\ref{q-set}) via  some $\S(\H_{A_1...A_n})$-valued function $\tilde{\omega}(x)$ and $\S_0$ be a convex subset of $\S(\H_{A_1...A_n})$ possessing property (\ref{S-prop}).\smallskip
}

\noindent A) \emph{If $f$ is a  function  in $\widehat{L}_n^{m}(C,D\vert\,\S_0)$ then
\begin{equation}\label{main+1}
    \vert f(\rho)-f(\sigma)\vert \leq C\varepsilon \ln d_m(\rho,\sigma)+Dh(\varepsilon)
\end{equation}
for any states $\rho$ and  $\sigma$ in $\mathfrak{Q}_{X,\F,\tilde{\omega}}\cap\S_0$ such that
\begin{equation}\label{v-cond}
 \mathrm{TV}(\mu_{\rho},\mu_{\sigma})=\varepsilon
\end{equation}
and $\,d_m(\rho,\sigma)\doteq\max\left\{\prod_{k=1}^{m}\rank\rho_{A_k}, \prod_{k=1}^{m}\rank\sigma_{A_k}\right\}$ is finite, where $\mu_{\rho}$ and $\mu_{\sigma}$ are measures in $\P(X)$ representing the states $\rho$ and $\sigma$.\smallskip
}

\emph{Condition (\ref{v-cond}) can be replaced with the condition $\mathrm{TV}(\mu_{\rho},\mu_{\sigma})\leq\varepsilon$
provided that the r.h.s. of (\ref{main+1}) is replaced with one of the expressions
\begin{itemize}
  \item $\{C\varepsilon \ln d_m(\rho,\sigma)+Dh(\varepsilon)\}^{\uparrow}_{\varepsilon}$
  \item $C\varepsilon \ln d_m(\rho,\sigma)+D h^{\uparrow}(\varepsilon)$,
\end{itemize}
where $\,\{...\}^{\uparrow}_{\varepsilon}$ is defined in (\ref{men-def+}) and $\,h^{\uparrow}$ is the function defined in (\ref{h+}).}
\smallskip

\noindent B) \emph{If $f$ is a nonnegative function in $\widehat{L}_n^{m}(C,D\vert\,\S_0)$ and $\rho$ is a state in $\,\mathfrak{Q}_{X,\F,\tilde{\omega}}\cap\S_0\,$ such that $\,d_m(\rho)\doteq\prod_{k=1}^{m}\rank\rho_{A_k}$ is finite then
\begin{equation}\label{main++1}
    f(\rho)-f(\sigma)\leq C\varepsilon \ln d_m(\rho)+Dh(\varepsilon)
\end{equation}
for any state $\sigma$ in $\mathfrak{Q}_{X,\F,\tilde{\omega}}\cap\S_0$ such that condition (\ref{v-cond}) holds.}
\emph{Inequality (\ref{main++1}) remains valid with $d_m(\rho)$ replaced by $\,d^*_m(\rho,\sigma)=\prod_{k=1}^{m}\rank([\rho-\sigma]_+)_{A_k}$, where $[\rho-\sigma]_+$ is the positive part of the Hermitian operator $\rho-\sigma$.}
\smallskip

\emph{Condition (\ref{v-cond}) can be replaced with the condition $\mathrm{TV}(\mu_{\rho},\mu_{\sigma})\leq\varepsilon$
provided that the r.h.s. of (\ref{main++1}) is replaced with one of the expressions
\begin{itemize}
  \item $\{C\varepsilon \ln d_m(\rho)+Dh(\varepsilon)\}^{\uparrow}_{\varepsilon}$;
  \item $C\varepsilon \ln d_m(\rho)+Dh^{\uparrow}(\varepsilon)$,
\end{itemize}
where $\,\{...\}^{\uparrow}_{\varepsilon}$ is defined in (\ref{men-def+}),  $\,h^{\uparrow}$ is the function defined in (\ref{h+}) and $d_m(\rho)$
can be replaced by $d^*_m(\rho,\sigma)$.}
\end{theorem}\smallskip


\begin{remark}\label{S-prop-r+2}
Both claims of Theorem \ref{main-1} remain valid if $\S_0$ is \emph{any} convex subset of $\S(\H_{A_1...A_n})$
provided that condition (\ref{S-prop+}) holds for the states $\rho$ and $\sigma$. This follows from
the proof of this theorem and Remark \ref{S-prop-r}.
\end{remark}\smallskip

\emph{Proof.} In the proofs of both parts of the theorem we may assume that $f$ is a function in $L_n^{m}(C,D\vert\,\S_0)$ . Indeed,  the expressions in r.h.s. of (\ref{main+1}) and  (\ref{main++1}) depend only on the parameters $C$ and $D$ and the characteristics of the states $\rho$ and $\sigma$.
\smallskip

A) Note first that the condition $d_m(\rho,\sigma)<+\infty$ implies that $f(\rho),f(\sigma)<+\infty$ by inequality (\ref{Cm}).
Since $\,\rank\varrho_{A_k}\leq\rank\rho_{A_k}$ and $\,\rank\varsigma_{A_k}\leq\rank\sigma_{A_k}$, $k=\overline{1,m}$,
for any states $\varrho$ and $\varsigma$ in $\mathfrak{Q}_{X,\F,\shs\tilde{\omega}}$ such that $\varepsilon\varrho\leq\rho$ and $\varepsilon\varsigma\leq\sigma$, inequality (\ref{Cm}) implies that
\begin{equation}\label{p-ineq-1}
\vert f(\varrho)-f(\varsigma)\vert \leq C\ln d_m(\rho,\sigma)
\end{equation}
for any such states $\varrho$ and $\varsigma$.  So, the quantities $C_f(\rho,\sigma\shs\vert \shs\varepsilon)$ and $C_f(\sigma,\rho\shs\vert \shs\varepsilon)$ defined in (\ref{C-f}) does not exceed
the r.h.s. of (\ref{p-ineq-1}). Thus, by applying Lemma \ref{g-ob} twice we obtain (\ref{main+1}).

The last claim of A obviously follows from the first one. It suffices only to note that
\begin{equation}\label{o-ineq}
\{C\varepsilon \ln d_m(\rho)+Dh(\varepsilon)\}_{\varepsilon}^{\uparrow}\leq C\varepsilon \ln d_m(\rho)+Dh^{\uparrow}(\varepsilon)
\end{equation}
because the function $h^{\uparrow}(\varepsilon)$ is non-decreasing on $[0,1]$.\smallskip

B) The condition $d_m(\rho)<+\infty$ implies that $f(\rho)<+\infty$ by inequality (\ref{Cm}).
So, the assertion follows directly from the last claim of Lemma \ref{g-ob}, since $\,\rank\varrho_{A_k}\leq\rank\rho_{A_k}<+\infty$, $k=\overline{1,m}$,
for any state $\varrho$ in $\mathfrak{Q}_{X,\F,\shs\tilde{\omega}}$ such that $\varepsilon\varrho\leq\rho$. Thus,
inequality (\ref{Cm}) implies that
\begin{equation}\label{p-ineq}
f(\varrho)\leq C\ln d_m(\rho)
\end{equation}
for any such state $\varrho$. So, the quantity $C^+_f(\rho\shs\vert \shs\varepsilon)$ defined in (\ref{C-f+}) does not exceed
the r.h.s. of (\ref{p-ineq}).

The possibility to replace
$d_m(\rho)$ with  $\,d^*_m(\rho,\sigma)$ in (\ref{main++1}) can be shown by upgrading the above arguments
with the use of inequality (\ref{AFW-1}) in the proof of Lemma \ref{g-ob} instead of (\ref{AFW-1+}) (see Remark \ref{S-prop-r+}).

The last claim of B obviously follows from the first one and (\ref{o-ineq}). $\Box$ \medskip

Now we apply Lemma \ref{g-ob} in Section 3.1 to obtain  continuity and semicontinuity bounds for functions from the classes
$\widehat{L}_n^{m}(C,D\vert\,\S_0)$ under the energy-type constraint  corresponding to the positive operator $H_m\otimes I_{A_{m+1}}\otimes...\otimes I_{A_n}$, where
\begin{equation*}
H_{m}=H_{A_1}\otimes I_{A_2}\otimes...\otimes I_{A_m}+\cdots+I_{A_1}\otimes... \otimes I_{A_{m-1}}\otimes H_{A_m}
\end{equation*}
is a positive operator on the space $\H_{A_1...A_m}$ determined by positive operators $H_{A_1}$,....,$H_{A_m}$ on the spaces $\H_{A_1}$,....,$\H_{A_m}$ (it is assumed that $H_1=H_{A_1}$). It is essential that the operator $H_{m}$ satisfies conditions (\ref{H-cond}) and (\ref{star}) if all the operators $H_{A_1}$,....,$H_{A_m}$  satisfy these conditions   \cite[Section 3.2.3]{QC}. Note that $\Tr H_{m}\rho=\sum_{k=1}^{m}\Tr H_{A_k}\rho_{A_k}$ for any $\rho$ in $\S(\H_{A_1..A_m})$.
We will use the function
\begin{equation}\label{F-H-m}
F_{H_{m}}(E)=\sup\shs\{ S(\rho)\,\vert \,\rho\in\S(\H_{A_1..A_m}),\,\Tr H_{m}\rho\leq E \}.
\end{equation}
If all the operators $H_{A_1}$,....,$H_{A_m}$  are isometrically  equivalent to each other then it is easy to see that $F_{H_{m}}(E)=mF_{H_{\!A_1}}(E/m)$, where $F_{H_{\!A_1}}$ is the function defined in (\ref{F-def}).\smallskip

The following theorem is an advanced version of Theorem 2 in \cite{LCB}.\smallskip

\begin{theorem}\label{main-2} \emph{Let $H_{\!A_1}$,...,$H_{\!A_m}$ be positive operators on the Hilbert spaces $\H_{A_1}$,....,$\H_{A_m}$ satisfying condition (\ref{H-cond})
and (\ref{star}) and $F_{H_{m}}$ the function defined in (\ref{F-H-m}). Let  $\mathfrak{Q}_{X,\F,\tilde{\omega}}$ be the set of states in $\S(\H_{A_1...A_n})$ defined in (\ref{q-set}) via  $\S(\H_{A_1...A_n})$-valued  function $\tilde{\omega}(x)$ and $\S_0$ a convex subset of $\S(\H_{A_1...A_n})$ with property (\ref{S-prop}).}

\smallskip

\noindent A) \emph{If $f$ is a function in $\widehat{L}_n^{m}(C,D\vert\,\S_0)$ then
\begin{equation}\label{main+2}
    \vert f(\rho)-f(\sigma)\vert \leq C\varepsilon F_{H_{m}}(mE/\varepsilon)+Dh(\varepsilon)
\end{equation}
for any states $\rho$ and $\sigma$ in $\,\mathfrak{Q}_{X,\F,\tilde{\omega}}\cap\S_0\,$ such that
\begin{equation}\label{v-cond+}
 \mathrm{TV}(\mu_{\rho},\mu_{\sigma})=\varepsilon
\end{equation}
and $\,\sum_{k=1}^{m}\Tr H_{A_k}\rho_{A_k},\,\sum_{k=1}^{m}\Tr H_{A_k}\sigma_{A_k}\leq mE$, where $\mu_{\rho}$ and $\mu_{\sigma}$ are measures in $\P(X)$ representing the states $\rho$ and $\sigma$.}\smallskip

\emph{Condition (\ref{v-cond+}) can be replaced with the condition $\,\mathrm{TV}(\mu_{\rho},\mu_{\sigma})\leq\varepsilon$
provided that the r.h.s. of (\ref{main+2}) is replaced with one of the expressions
\begin{itemize}
  \item $\{C\varepsilon F_{H_{m}}(mE/\varepsilon)+Dh(\varepsilon)\}^{\uparrow}_{\varepsilon}$
  \item $C\varepsilon F_{H_{m}}(mE/\varepsilon)+Dh^{\uparrow}(\varepsilon)$,
\end{itemize}
where $\,\{...\}^{\uparrow}_{\varepsilon}$ is defined in (\ref{men-def+}) and  $\,h^{\uparrow}$ is the function defined in (\ref{h+}).}\smallskip

\noindent B) \emph{If $f$ is a nonnegative function in $\widehat{L}_n^{m}(C,D\vert\,\S_0)$ and $\rho$ is a state in $\,\mathfrak{Q}_{X,\F,\tilde{\omega}}\cap\S_0\,$ such that $\, E(\rho)\doteq(1/m)\sum_{k=1}^{m}\Tr H_{A_k}\rho_{A_k}<+\infty\,$  then
\begin{equation}\label{main++2}
    f(\rho)-f(\sigma)\leq C\varepsilon F_{H_{m}}(mE(\rho)/\varepsilon)+Dh(\varepsilon)
\end{equation}
for any state $\sigma$ in $\mathfrak{Q}_{X,\F,\tilde{\omega}}\cap\S_0$ such that condition (\ref{v-cond+}) holds. Inequality (\ref{main++2}) remains valid with $E(\rho)$ replaced by $\,E(\rho,\sigma)\doteq(1/m)\sum_{k=1}^{m}\Tr H_{A_k}([\rho-\sigma]_+)_{A_k}$, where $[\rho-\sigma]_+$ is the positive part of the Hermitian operator
$\rho-\sigma$. If the set $\mathfrak{Q}_{X,\F,\tilde{\omega}}$ consists of commuting states then  (\ref{main++2}) holds with $E(\rho)$ replaced by $\,E(\rho)-E_{\varepsilon}(\rho)$, where $E_{\varepsilon}(\rho)\doteq(1/m)\sum_{k=1}^{m}\Tr H_{A_k}([\rho-\varepsilon I_{A_1...A_n}]_+)_{A_k}$.}

\emph{Condition (\ref{v-cond+}) can be replaced with the condition $\,\mathrm{TV}(\mu_{\rho},\mu_{\sigma})\leq\varepsilon$
provided that the r.h.s. of (\ref{main++2}) is replaced with one of the expressions
\begin{itemize}
  \item $\{C\varepsilon F_{H_{m}}(mE(\rho)/\varepsilon)+Dh(\varepsilon)\}^{\uparrow}_{\varepsilon}$;
  \item $C\varepsilon F_{H_{m}}(mE(\rho)/\varepsilon)+Dh^{\uparrow}(\varepsilon)$,
\end{itemize}
where $\,\{...\}^{\uparrow}_{\varepsilon}$ is defined in (\ref{men-def+}),  $\,h^{\uparrow}$ is the function defined in (\ref{h+}) and $E(\rho)$
can be replaced by $E(\rho,\sigma)$. If the set $\mathfrak{Q}_{X,\F,\tilde{\omega}}$ consists of commuting states then $E(\rho)$ in the above expressions
can be replaced by $\,E(\rho)-E_{\varepsilon}(\rho)$.}
\end{theorem}\smallskip

\begin{remark}\label{main-r-1}
Since the operator $H_m$ satisfies conditions (\ref{H-cond}), the equivalence of (\ref{H-cond}) and (\ref{H-cond-a}) show that the
r.h.s. of (\ref{main+2}) and  (\ref{main++2})  tends to zero as $\,\varepsilon\to0$.
\end{remark}\smallskip


\begin{remark}\label{S-prop-r++}
Both claims of Theorem \ref{main-2} remain valid if $\S_0$ is \emph{any} convex subset of $\S(\H_{A_1...A_n})$
provided that condition (\ref{S-prop+}) holds for the states $\rho$ and $\sigma$. This follows from
the proof of this theorem and Remark \ref{S-prop-r}.
\end{remark}\smallskip

\emph{Proof.}  In the proofs of both parts of the theorem we may assume that $f$ is a function in $L_n^{m}(C,D\vert\,\S_0)$. Indeed, the expressions in r.h.s. of (\ref{main+2}) and  (\ref{main++2}) depend only on the parameters $C$ and $D$ and the characteristics of the states $\rho$ and $\sigma$.\smallskip

A) Since $\Tr H_{m}[\vartheta_{A_1}\otimes...\otimes\vartheta_{A_m}]=\sum_{k=1}^m \Tr H_{A_k}\vartheta_{A_k}$, we have
$$
\sum_{k=1}^m S(\vartheta_{A_k})=S(\vartheta_{A_1}\otimes...\otimes\vartheta_{A_m})\leq F_{H_{m}}(mE)
$$
for any state $\vartheta\in \S_0$ such that $\Tr H_{m}\vartheta_{A_1..A_m}=\sum_{k=1}^m\Tr H_{A_k}\vartheta_{A_k}\leq mE$. Hence, for any such state $\vartheta$  inequality (\ref{Cm}) implies that
\begin{equation}\label{F-p-2+}
-c_f^-F_{H_{m}}(mE)\leq f(\vartheta)\leq c_f^+F_{H_{m}}(mE).
\end{equation}
It follows, in particular, that $f(\rho)$ and $f(\sigma)$ are finite.

Assume that $\varrho$ and $\varsigma$ are states in $\mathfrak{Q}_{X,\F,\shs\tilde{\omega}}\cap\S_0$ such that $\varepsilon\varrho\leq\rho$ and $\varepsilon\varsigma\leq\sigma$. Then
\begin{equation}\label{en-est}
\varepsilon\sum_{k=1}^m \Tr H_{A_k}\varrho_{A_k}\leq\sum_{k=1}^m \Tr H_{A_k}\rho_{A_k}\leq mE
\end{equation}
and
$$
\varepsilon\sum_{k=1}^m \Tr H_{A_k}\varsigma_{A_k}\leq\sum_{k=1}^m \Tr H_{A_k}\sigma_{A_k}\leq mE.
$$
Thus, it follows from (\ref{F-p-2+}) that
$$
\vert f(\varrho)-f(\varsigma)\vert \leq CF_{H_{m}}(mE/\varepsilon).
$$
So, the quantities $C_f(\rho,\sigma\shs\vert \shs\varepsilon)$ and $C_f(\sigma,\rho\shs\vert \shs\varepsilon)$ defined in (\ref{C-f})  do not exceed
$CF_{H_{m}}(mE/\varepsilon)$. Thus, by applying Lemma \ref{g-ob} twice we obtain (\ref{main+2}).

The last claim of A  follows from the first one. It suffices to note that
\begin{equation}\label{o-ineq+}
\{C\varepsilon F_{H_{m}}(mE/\varepsilon)+Dh(\varepsilon)\}^{\uparrow}_{\varepsilon}\leq C\varepsilon F_{H_{m}}(mE/\varepsilon)+Dh^{\uparrow}(\varepsilon)
\end{equation}
because the functions  $\varepsilon\mapsto \varepsilon F_{H_{m}}(mE/\varepsilon)$ and $h^{\uparrow}(\varepsilon)$ are non-decreasing on $[0,1]$ (the
first function is non-decreasing by Lemma \ref{W-L}, since the function $F_{H_{m}}$ is concave on $\mathbb{R}_+$).\smallskip

B) The main assertion of part B follows from the last claim of Lemma \ref{g-ob}, since the arguments from the proof of part A show
that $f(\rho)<+\infty$ and that
\begin{equation}\label{p-ineq-2}
f(\varrho)\leq CF_{H_{m}}(mE(\rho)/\varepsilon)
\end{equation}
for any state $\varrho$ in $\mathfrak{Q}_{X,\F,\shs\tilde{\omega}}\cap\S_0$ such that $\varepsilon\varrho\leq\rho$.
So, the quantity $C^+_f(\rho\shs\vert \shs\varepsilon)$ defined in (\ref{C-f+}) does not exceed
the r.h.s. of (\ref{p-ineq-2}).

The possibility to replace
$E(\rho)$ by $\,E(\rho,\sigma)$ in (\ref{main++2}) can be shown by upgrading the above arguments
with the use of inequality (\ref{AFW-1}) in the proof of Lemma \ref{g-ob} instead of (\ref{AFW-1+}) (see Remark \ref{S-prop-r+}).

If a state $\varrho$ commutes with the state $\rho$ then the condition $\,\varepsilon\varrho\leq\rho$
implies that
$$
\varepsilon\varrho\leq \rho\wedge\varepsilon I_{A_1...A_n}=\rho-[\rho-\varepsilon I_{A_1...A_n}]_+,
$$
since $\varrho\leq I_{A_1...A_n}$. By incorporating this  to the estimates in (\ref{en-est}) it is easy to prove the corresponding claim of part B.\smallskip

The last claim of B  follows from the first one and (\ref{o-ineq+}). $\Box$\smallskip

Theorems \ref{main-1} and \ref{main-2} allow us to refine the continuity bounds for functions from the classes $\widehat{L}_n^{m}(C,D)$
presented in \cite[Section 4]{QC} in the case of commuting states, i.e. such states $\rho$ and $\sigma$ in $\,\S(\H_{A_1..A_n})$ that
$[\rho,\sigma]\doteq \rho\,\sigma-\sigma\rho=0$.
Indeed, since for any commuting states $\rho$ and $\sigma$ in $\,\S(\H_{A_1..A_n})$ there exists an orthonormal
basis $\{\vert k\rangle\}_{k\in\N}$ in $\H_{A_1..A_n}$ in which these states are diagonizable, they belong to the set
$\mathfrak{Q}_{X,\F,\tilde{\omega}}$, where $X=\N$, $\F$ is the $\sigma$-algebra of all subsets of $\N$ and $\tilde{\omega}(k)=\vert k\rangle\langle k\vert$.
In this case the set $\P(X)$ can be identified with the set of probability distributions on $\N$, the representing measures $\mu_{\rho}$ and $\mu_{\sigma}$
correspond to the probability distributions formed by the eigenvalues of $\rho$ and $\sigma$  and
$\mathrm{TV}(\mu_{\rho},\mu_{\sigma})=\textstyle\frac{1}{2}\|\rho-\sigma\|_1$.\smallskip

Thus, Theorem \ref{main-1} implies the following\smallskip

\begin{corollary}\label{main-1-c} \emph{Let  $\S_0$ be a convex subset of $\,\S(\H_{A_1...A_n})$ such that}
\begin{equation}\label{S-prop+++}
 \{\varrho\in\S_0\}\wedge\{\exists c>0:c\varsigma\leq\varrho \}\wedge\{[\varrho,\varsigma]=0\}\;\Rightarrow\;\varsigma\in\S_0.
\end{equation}

\noindent A) \emph{If $f$ is a  function  in $\widehat{L}_n^{m}(C,D\vert\,\S_0)$  then
\begin{equation}\label{main+1-c}
    \vert f(\rho)-f(\sigma)\vert \leq C\varepsilon \ln d_m(\rho,\sigma)+Dh(\varepsilon)
\end{equation}
for any states $\rho$ and  $\sigma$ in $\,\S_0$ s.t. $\,d_m(\rho,\sigma)\doteq\max\left\{\prod_{k=1}^{m}\rank\rho_{A_k}, \prod_{k=1}^{m}\rank\sigma_{A_k}\right\}$ is finite, $\,[\rho,\sigma]=0\,$ and
\begin{equation}\label{v-cond++}
 \frac{1}{2}\|\rho-\sigma\|_1=\varepsilon.
\end{equation}
If $\,m=n=1$ then (\ref{main+1-c}) holds with $d_m(\rho,\sigma)$ replaced by $\,d_*(\rho,\sigma)-1$, where $d_*(\rho,\sigma)$ is the dimension of the minimal subspace containing the supports of
$\rho$ and $\sigma$.}\smallskip

\emph{Condition (\ref{v-cond++}) can be replaced with the condition $\frac{1}{2}\|\rho-\sigma\|_1\leq\varepsilon$
provided that the r.h.s. of (\ref{main+1-c}) is  replaced with one of the expressions\pagebreak
\begin{itemize}
  \item $\{C\varepsilon \ln d_m(\rho,\sigma)+Dh(\varepsilon)\}^{\uparrow}_{\varepsilon}$
  \item $C\varepsilon \ln d_m(\rho,\sigma)+Dh^{\uparrow}(\varepsilon)$,
\end{itemize}
where $\,\{...\}^{\uparrow}_{\varepsilon}$ is defined in (\ref{men-def+}) and  $\,h^{\uparrow}$ is the function defined in (\ref{h+}). If $\,m=n=1$ then  $d_m(\rho,\sigma)$ in the above expressions can be replaced by $\,d_*(\rho,\sigma)-1$.}\smallskip

\noindent B) \emph{If $f$ is a nonnegative function in $\widehat{L}_n^{m}(C,D\vert\,\S_0)$ and $\rho$ is a state in $\,\S_0$ such that $\,d_m(\rho)\doteq\prod_{k=1}^{m}\rank\rho_{A_k}$ is finite then
\begin{equation}\label{main++1-c}
    f(\rho)-f(\sigma)\leq C\varepsilon\ln d_m(\rho)+Dh(\varepsilon)
\end{equation}
for any state $\sigma$ in $\,\S_0$ such that  $\,[\rho,\sigma]=0\,$ and condition (\ref{v-cond++}) holds.}
\emph{Inequality (\ref{main++1-c}) remains valid with $d_m(\rho)$ replaced by $\,d^*_m(\rho,\sigma)=\prod_{k=1}^{m}\rank([\rho-\sigma]_+)_{A_k}$, where $[\rho-\sigma]_+$ is the positive part of the Hermitian operator
$\rho-\sigma$.}\smallskip

\emph{Condition (\ref{v-cond++}) can be replaced with the condition $\frac{1}{2}\|\rho-\sigma\|_1\leq\varepsilon$
provided that the r.h.s. of (\ref{main++1-c}) is replaced with one of the expressions
\begin{itemize}
  \item $\{C\varepsilon \ln d_m(\rho)+Dh(\varepsilon)\}^{\uparrow}_{\varepsilon}$
  \item $C\varepsilon \ln d_m(\rho)+Dh^{\uparrow}(\varepsilon)$,
\end{itemize}
where $\,\{...\}^{\uparrow}_{\varepsilon}$ is defined in (\ref{men-def+}),  $\,h^{\uparrow}$ is the function defined in (\ref{h+}) and $d_m(\rho)$ can be replaced by $d^*_m(\rho,\sigma)$.}

\end{corollary}

\begin{proof} Note that property (\ref{S-prop+++}) guarantees the validity of
condition  (\ref{S-prop+}) for any commuting states $\rho$ and $\sigma$ in $\S_0$.
So, by the observation before the corollary the main part of its first claim and the second claim  follow directly from Theorem \ref{main-1}
and  Remark \ref{S-prop-r}.


The claim concerning  the case $m=n=1\,$ is proved by using  Lemma \ref{g-ob}  with the settings mentioned before Corollary \ref{main-1-c}.
Indeed, it is not hard to see that in this case  the rank of the states $\Omega(\mu)$ and $\Omega(\nu)$
in the definition (\ref{C-f}) of $C_f(\rho,\sigma\vert\shs\varepsilon)$  is
\emph{strictly less } than $d_*(\rho,\sigma)$. By using this observation it is easy to show that
\begin{equation*}
\begin{array}{c}
C_f(\rho,\sigma\vert\shs\varepsilon)\leq\sup\left\{f(\varrho)-f(\varsigma)\,\left\vert\, \varrho,\varsigma\in\S(\H): \rank\varrho,\rank\varsigma \leq d_*(\rho,\sigma)-1\right.\right\}\\\\\leq C\ln(d_*(\rho,\sigma)-1),
\end{array}
\end{equation*}
where the last inequality follows from the definition of the class $\widehat{L}_1^{1}(C,D\vert\,\S_0)$. The same estimate holds for $C_f(\sigma,\rho\vert\shs\varepsilon)$
\end{proof}
\smallskip\pagebreak

Theorem \ref{main-2} implies the following\smallskip

\begin{corollary}\label{main-c} \emph{Let $\S_0$ be a convex subset of $\,\S(\H_{A_1...A_n})$ with property (\ref{S-prop+++}). Let $H_{\!A_1}$,...,$H_{\!A_m}$ be positive operators on the Hilbert spaces $\H_{A_1}$,....,$\H_{A_m}$ satisfying conditions (\ref{H-cond})
and (\ref{star}) and $F_{H_{m}}$ the function defined in (\ref{F-H-m}).}\smallskip

\noindent A) \emph{If $f$ is a function in $\widehat{L}_n^{m}(C,D\vert\,\S_0)$ then
\begin{equation}\label{main+c}
    \vert f(\rho)-f(\sigma)\vert \leq C\varepsilon F_{H_{m}}(mE/\varepsilon)+Dh(\varepsilon)
\end{equation}
for any  states $\rho$ and $\sigma$ in $\,\S_0$ such that $\,\sum_{k=1}^{m}\Tr H_{A_k}\rho_{A_k},\,\sum_{k=1}^{m}\Tr H_{A_k}\sigma_{A_k}\leq mE\,$,
$\,[\rho,\sigma]=0$ and
\begin{equation}\label{v-cond+2}
 \frac{1}{2}\|\rho-\sigma\|_1=\varepsilon.
\end{equation}
}\smallskip
\emph{Condition (\ref{v-cond+2}) can be replaced with the condition $\,\frac{1}{2}\|\rho-\sigma\|_1\leq\varepsilon$
provided that the r.h.s. of (\ref{main+c}) is replaced with one of the expressions
\begin{itemize}
  \item $\{C\varepsilon F_{H_{m}}(mE/\varepsilon)+Dh(\varepsilon)\}^{\uparrow}_{\varepsilon}$
  \item $C\varepsilon F_{H_{m}}(mE/\varepsilon)+Dh^{\uparrow}(\varepsilon)$,
\end{itemize}
where $\,\{...\}^{\uparrow}_{\varepsilon}$ is defined in (\ref{men-def+}) and $\,h^{\uparrow}$ is the function defined in (\ref{h+}).}
\smallskip

\noindent B) \emph{If $f$ is a nonnegative function in $\widehat{L}_n^{m}(C,D\vert\,\S_0)$ and $\rho$ is a state in $\,\S_0$ such that $\,E(\rho)\doteq(1/m)\sum_{k=1}^{m}\Tr H_{A_k}\rho_{A_k}<+\infty\,$  then
\begin{equation}\label{main++c}
    f(\rho)-f(\sigma)\leq C\varepsilon F_{H_{m}}(m(E(\rho)-E_{\varepsilon}(\rho))/\varepsilon)+Dh(\varepsilon)
\end{equation}
and
\begin{equation}\label{main+++c}
    f(\rho)-f(\sigma)\leq C\varepsilon F_{H_{m}}(mE(\rho)/\varepsilon)+Dh(\varepsilon)
\end{equation}
for any state $\sigma$ in $\,\S_0$ such that $\,[\rho,\sigma]=0$ and condition (\ref{v-cond+2}) holds,
where\footnote{$\,[\rho-\varepsilon I_{A_1...A_n}]_+$ is the positive part of the Hermitian operator $\,\rho-\varepsilon I_{A_1...A_n}$.}
$$
E_{\varepsilon}(\rho)=\frac{1}{m}\sum_{k=1}^{m}\Tr H_{A_k}\langle[\rho-\varepsilon I_{A_1...A_n}]_+\rangle_{A_k}.
$$
Inequality (\ref{main+++c}) is valid with $E(\rho)$ replaced by $\,E(\rho,\sigma)\doteq(1/m)\sum_{k=1}^{m}\Tr H_{A_k}([\rho-\sigma]_+)_{A_k}$, where $[\rho-\sigma]_+$ is the positive part of the Hermitian operator $\rho-\sigma$.}\smallskip

\emph{Condition (\ref{v-cond+2}) can be replaced with the condition $\,\frac{1}{2}\|\rho-\sigma\|_1\leq\varepsilon$
provided that the r.h.s. of (\ref{main++c}) (resp. of (\ref{main+++c})) is replaced with one of the expressions
\begin{itemize}
  \item $\{ C\varepsilon F_{H_{m}}(m(E(\rho)-E_{\varepsilon}(\rho))/\varepsilon)+Dh(\varepsilon)\}_{\varepsilon}^{\uparrow}$ (resp. $\{ C\varepsilon F_{H_{m}}((mE_*/\varepsilon)+Dh(\varepsilon)\}_{\varepsilon}^{\uparrow}$)
  \item $C\varepsilon F_{H_{m}}(m(E(\rho)-E_{\varepsilon}(\rho))/\varepsilon)+Dh^{\uparrow}(\varepsilon)_{\varepsilon}$ (resp. $C\varepsilon F_{H_{m}}((mE_*/\varepsilon)+Dh^{\uparrow}(\varepsilon)$),
\end{itemize}
where $\,\{...\}^{\uparrow}_{\varepsilon}$ is defined in (\ref{men-def+}), $\,h^{\uparrow}$ is the function defined in (\ref{h+}) and $E_*$ is  either  $\,E(\rho)$ or $E(\rho,\sigma)$.}\end{corollary}

\begin{proof} By the observation before Corollary \ref{main-1-c} all the claims of this corollary follow from Theorem \ref{main-2}
and  Remark \ref{S-prop-r++}. It suffices to note that property (\ref{S-prop+++}) guarantees the validity of
condition  (\ref{S-prop+}) for any commuting states $\rho$ and $\sigma$ in $\S_0$.
\end{proof}\smallskip

\noindent\textbf{Note:} Since $E_{\varepsilon}(\rho)$ tends to $\,E(\rho)\doteq(1/m)\sum_{k=1}^{m}\Tr H_{A_k}\rho_{A_k}$ as $\,\varepsilon\to0$, the estimate (\ref{main++c}) is essentially  sharper than (\ref{main+++c}) provided that $\varepsilon\ll1$.\medskip

\begin{example}\label{c-1-e}  Let $f(\rho)=I(A_1\!:\!A_2)_{\rho}$ be the quantum mutual
information defined in (\ref{mi-d}). This is a function on the whole space $\S(\H_{A_1A_2})$  taking values in $[0,+\infty]$ and satisfying the inequalities (\ref{MI-LAA-1}) and (\ref{MI-LAA-2}) with possible value $+\infty$ in both sides. This and upper bound (\ref{MI-UB}) show that $I(A_1\!:\!A_2)$ belongs to the class
$L_2^1(2,2\vert \S(\H_{A_1A_2}))$.

Thus, Corollary \ref{main-1-c}B implies that
\begin{equation}\label{I-1-1}
I(A_1\!:\!A_2)_\rho-I(A_1\!:\!A_2)_\sigma\leq \{2\varepsilon \ln \rank\rho_{A_1}+2h(\varepsilon)\}^{\uparrow}_{\varepsilon}\leq 2\varepsilon \ln \rank\rho_{A_1}+2h^{\uparrow}(\varepsilon)
\end{equation}
for any commuting states $\rho$ and $\sigma$ in $\,\S(\H_{A_1A_2})$ such that $\frac{1}{2}\|\rho-\sigma\|_1\leq \varepsilon$,
where $\,\{...\}^{\uparrow}_{\varepsilon}$ is defined in (\ref{men-def+}) and $\,h^{\uparrow}$ is the function defined in (\ref{h+})

Corollary \ref{main-c}B implies that
\begin{equation}\label{I-1-2}
I(A_1\!:\!A_2)_\rho-I(A_1\!:\!A_2)_\sigma\leq 2\varepsilon F_{H}(E/\varepsilon)+2h^{\uparrow}(\varepsilon)
\end{equation}
for any state $\rho$ in $\,\S(\H_{A_1A_2})$ such that $\,E\doteq\Tr H\rho_{A_1}<+\infty$ and arbitrary state  $\sigma$ in $\,\S(\H_{A_1A_2})$ such that
$[\rho,\sigma]=0$ and $\frac{1}{2}\|\rho-\sigma\|_1\leq \varepsilon$, where $H$ is any positive operator on $\H_{A_1}$ satisfying conditions (\ref{H-cond})
and (\ref{star}) and $\,h^{\uparrow}$ is the function defined in (\ref{h+}). The semicontinuity bound (\ref{I-1-2}) can be refined by replacing $E$ with $E-\Tr H\langle[\rho-\varepsilon I_{A_1A_2}]_+\rangle_{A_1}$.

It follows from (\ref{I-1-2}) that
\begin{equation}\label{I-2-2}
\vert I(A_1\!:\!A_2)_\rho-I(A_1\!:\!A_2)_\sigma\vert \leq 2\varepsilon F_{H}(E/\varepsilon)+2h^{\uparrow}(\varepsilon)
\end{equation}
for any commuting states $\rho$ and $\sigma$ in $\,\S(\H_{A_1A_2})$ such that  $\,\Tr H\rho_{A_1},\Tr H\sigma_{A_1}\leq E\,$ and $\frac{1}{2}\|\rho-\sigma\|_1\leq \varepsilon$,
where $H$ is any positive operator on $\H_{A_1}$ satisfying conditions (\ref{H-cond})
and (\ref{star}).

Comparing continuity bound (\ref{I-2-2}) with the universal continuity bound for $I(A_1\!:\!A_2)$ under the energy constraint presented in \cite[Proposition 14]{QC}
we see an \emph{essential gain in accuracy} (not to mention mathematical simplicity). Of course, this gain is due to the fact that we restrict attention to commuting states.

The semicontinuity bounds (\ref{I-1-1}) and (\ref{I-1-2}) remain valid for the  quantum conditional mutual information
$I(A_1\!:\!A_2|A_3)$ (extended to the set of all states of infinite-dimensional system $A_1A_2A_3$ by the rule described in
\cite[Section 4.2]{QC}) because the function  $f(\rho)=I(A_1\!:\!A_2|A_3)_{\rho}$ is nonnegative and belongs to the class $L_3^1(2,2\vert \S(\H_{A_1A_2A_3}))$.\smallskip

By using  Corollaries  \ref{main-1-c}B and \ref{main-c}B one can obtain the semicontinuity bounds similar to  (\ref{I-1-1}) and (\ref{I-1-2}) for the
following correlation and entanglement measures in a bipartite quantum system $A_1A_2$:  the relative entropy of entanglement, the relative entropy of NPT entanglement, the Rains bound, the one way classical correlation with measured system $A_1$, the one way classical correlation with measured system $A_2$, the quantum discord with measured system $A_1$, the quantum discord with measured system $A_2$, since they belong, respectively, to the classes $L_2^1(1,1)$, $L_2^1(1,1)$, $L_2^1(1,1)$, $\widehat{L}_2^1(1,2)$, $\widehat{L}_2^1(1,2)$, $\widehat{L}_2^1(1,2)$, $\widehat{L}_2^1(2,2)$ (see \cite[Section IV, the proof of Proposition 39]{L&Sh-1},\cite[Section 4]{QC}).
\end{example}

\section{Applications to characteristics of quantum systems}

\subsection{The von Neumann entropy}


The aim of this section is to show that the universal results of Section 3 allow us to reprove
optimal semicontinuity bounds for the von Neumann entropy in  both finite-dimensional and infinite-dimensional settings  (proved before by very specific methods)
and to obtain several new results.

Note first that that direct application of Lemma \ref{g-ob} to the function $f=S$ shows that the inequality
\begin{equation}\label{S-1}
S(\rho)+\varepsilon S(\tau_-)\leq S(\sigma)+\varepsilon S(\tau_+)+h(\varepsilon)
\end{equation}
holds  for any commuting\footnote{It is proved recently in \cite{D++} (by applying a quite sophisticated technique) that the inequality
(\ref{S-1}) holds regardless of the condition $[\rho,\sigma]=0$.} states $\rho$ and $\sigma$ in $\S(\H)$ provided that
$S(\rho)<+\infty$ and $\varepsilon=\,\frac{1}{2}\|\rho-\sigma\|_1$, where
$\tau_\pm=(1/\varepsilon)[\rho-\sigma]_\pm$ and $h$ is the binary entropy. Indeed,
it suffices to note the von Neumann entropy satisfies inequalities (\ref{LAA-1}) and (\ref{LAA-2})
with $a_f=0$ and $b_f=h$ and to apply inequality (\ref{AFW-1+}) from the proof of Lemma \ref{g-ob} by using
the observation before Corollary \ref{main-1-c} in Section 3.2.

In \cite{FCB} it is shown how the inequality
(\ref{S-1}) (more precisely, a simple corollary of this inequality formulated in terms of probability distributions and the Shannon entropy) can be used to prove
the following theorem whose part A (resp. B) with  $m=1$ provides the optimal semicontinuity bound for the von Neumann entropy with the rank (resp. energy) constraint originally obtained in \cite{LCB} (resp. in \cite{BDJSH}) by  different methods.
\smallskip

\begin{theorem}\label{SCB-m} \cite{FCB}\footnote{The claim of part A of this theorem  with $m=1$ coincides with the claim of Proposition 2 in \cite{LCB}, while
the claim of part B with $m=1$ coincides with the claim of Theorem 1 in \cite{BDJSH}.}
\emph{Let $\rho$ be a state in $\S(\H)$ and $m\in\N$ be arbitrary.}\smallskip

A)  \emph{If $\,m<d\doteq\rank\rho<+\infty$  then
\begin{equation*}
   S(\rho)-S(\sigma)\leq \left\{\begin{array}{lr}
        \varepsilon\ln(d-m)+h(\varepsilon)&\textrm{if}\;\;  \varepsilon\leq 1-1/(d-m+1)\; \\
        \ln (d-m+1)& \textrm{if}\;\; \varepsilon\geq 1-1/(d-m+1),
        \end{array}\right.
\end{equation*}
for any  state $\sigma$ in $\S(\H)$ such that  $\,\frac{1}{2}\|\rho-\sigma\|_1\leq\varepsilon\in[0,1]$  and
\begin{equation}\label{mj++}
\sum_{i=1}^{r}\lambda^{\sigma}_i\leq\sum_{i=1}^{r}\lambda^{\rho}_i\quad \forall r=1,2,..., m-1,
\end{equation}
where  $\{\lambda^{\rho}_i\}_{i=1}^{+\infty}$ and $\{\lambda^{\sigma}_i\}_{i=1}^{+\infty}$ are the  sequences of eigenvalues of the states $\rho$ and $\sigma$ arranged in the non-increasing order and it is assumed that the condition (\ref{mj++}) holds trivially for $m=1$.}
\smallskip

B)  \emph{If $E\doteq\Tr H\rho<+\infty$ for some
positive operator  $H$ with  representation (\ref{H-form}) satisfying condition (\ref{H-cond}) then
\begin{equation}\label{S-CB-B}
 S(\rho)-S(\sigma)\leq\left\{\begin{array}{l}
       \!\varepsilon F_{H_m}(E_m/\varepsilon)+h(\varepsilon)\quad\;\, \textrm{if}\;\;  \varepsilon\in[0,a_{H^0_m}\!(E_m)]\\\\
        \!F_{H^0_m}(E_m)\qquad \qquad\qquad\textrm{if}\;\;  \varepsilon\in[a_{H^0_m}\!(E_m),1]
        \end{array}\right.
\end{equation}
for any  state $\sigma$ in $\S(\H)$ such that  $\,\frac{1}{2}\|\rho-\sigma\|_1\leq\varepsilon$ and condition (\ref{mj++}) holds, where $E_m\doteq E-\sum_{i=1}^mh_i\lambda^{\rho}_i$,
$a_{H^0_m}\!(E_m)=1-1/Z_{H^0_m}\!(E_m)$,  $F_{H_m}$, $F_{H^0_m}$  and $Z_{H^0_m}$ are  the functions\footnote{The method of finding  these functions  for a given  operator $H$ is described after Theorem 1 in \cite{FCB}.} defined in (\ref{F-def})  and (\ref{Z}) with $H$ replaced by the operators
\begin{equation}\label{H-rep+}
H_m\doteq \sum_{i=m+1}^{+\infty} h_{i}|\tau_i\rangle\langle \tau_i|,\quad\textit{ and }\quad H^0_m=0|\tau_1\rangle\langle \tau_1|+\sum_{i=m+1}^{+\infty} h_{i}|\tau_i\rangle\langle \tau_i|,
\end{equation}
with  the domains $\;\mathcal{D}(H_m)=\{\varphi\in\H_{\mathcal{T}_m}\,|\,\sum_{i=m+1}^{+\infty} h^2_{i}|\langle\tau_i|\varphi\rangle|^2<+\infty\}$ and
$\mathcal{D}(H^0_m)=\{\varphi\in\H_{\mathcal{T}_m}\oplus\H_{\tau_1}\,|\,\sum_{i=m+1}^{+\infty} h^2_{i}|\langle\tau_i|\varphi\rangle|^2<+\infty\}$,
$\H_{\mathcal{T}_m}$ is the linear span of the system $\mathcal{T}_m\doteq\left\{\tau_i\right\}_{i=m+1}^{+\infty}$, $\H_{\tau_1}$ is the linear span of the vector $\tau_1$.}
\end{theorem}\smallskip

\begin{remark}\label{sv+}
The r.h.s. of (\ref{S-CB-B}) tends to zero as $\,\varepsilon\to0$ due to the equivalence of (\ref{H-cond}) and (\ref{H-cond-a}), since it is easy to see that the operator $H_m$ satisfies condition (\ref{H-cond}) (because $H$ satisfies this condition).
\end{remark}\smallskip

\begin{remark}\label{SCB-r} For any $\varepsilon\in(0,1]$ the r.h.s. of (\ref{S-CB-B}) can be replaced by
$$
\varepsilon F_{H^0_m}(E_m/\varepsilon)+h^{\uparrow}(\varepsilon).
$$
This can be shown by using Lemma \ref{W-L} and by noting  that $F_{H^m}(E)\leq F_{H^0_m}(E)$ for any $E\geq h_{m+1}$. If $m=1$ and the operator $H$ satisfies condition
(\ref{star}) (i.e. $h_1=0$)  then  the function $F_{H^0_m}$ in the above expression coincides with the function $F_{H}$, since $H^0_1=H$ in this case.
\end{remark}
\smallskip

The condition (\ref{mj++}) is interpreted in \cite{FCB} as \emph{$m$-partial majorization} of the state $\sigma$ by the state $\rho$. If the inequality
in (\ref{mj++}) holds for any $r$ (or, at least, for all $r<\rank\rho$ if $\rank\rho<+\infty$) then the state $\rho$ majorizes the state $\sigma$
in the conventional sense and, hence, the Schur concavity of the von Neumann  entropy implies that $S(\rho)\leq S(\sigma)$ \cite{M-2,N&Ch}. So, the claims of Theorem \ref{SCB-m}B with $m>1$ provides estimates  the degree of violation of the inequality $S(\rho)\leq S(\sigma)$ by states $\rho$ and $\sigma$ close to each other w.r.t. the trace norm for which
the majorization condition  is\emph{ only partially fulfilled}  (i.e. fulfilled in the sense of (\ref{mj++})).\medskip

The semicontinuty  bounds in Theorem \ref{SCB-m}B are applicable
to any state  $\rho$ in $\S(\H)$ with finite $S(\rho)$, since for any such state there is a positive operator $H$ on $\H$ satisfying conditions (\ref{H-cond}) and (\ref{star}) such that $\Tr H\rho<+\infty$ \cite[Proposition 4]{EC}. A simple lower semicontinuty bound  applicable
to any state  $\rho$ in $\S(\H)$ with finite $S(\rho)$ can be obtained by using Lemma \ref{b-lemma}A in Section 3. \smallskip

\begin{proposition}\label{S-SCB-2}\emph{If $\rho$ is  a state in $\S(\H)$ with finite entropy then
\begin{equation}\label{S-SCB-2+}
S(\rho)-S(\sigma)\leq \tilde{S}(\rho\wedge\varepsilon I_{\H})+h^{\uparrow}(\varepsilon)
\end{equation}
for any state $\sigma$ in $\,\S(\H)$ such that  $\,\frac{1}{2}\|\rho-\sigma\|_1\leq \varepsilon\in(0,1]$, where $\rho\wedge\varepsilon I_{\H}$ is the operator defined in (\ref{2-op}), $\tilde{S}$ is the extended von Neumann entropy defined in (\ref{S-ext}) and
$\,h^{\uparrow}$ is the function defined in (\ref{h+}).}\smallskip

\emph{The semicontinuity  bounds (\ref{S-SCB-2+}) is faithful: its right hand side tends to zero  as $\,\varepsilon\to0$.}

\end{proposition}\smallskip

To derive the main claim of this proposition from Lemma \ref{b-lemma}A it suffices to note the von Neumann entropy satisfies inequalities (\ref{LAA-1}) and (\ref{LAA-2})
with $a_f=0$ and $b_f=h$ and to repeat the arguments from the proof of Proposition 1 in \cite{LLB} (based on applying Mirsky's inequality (\ref{Mirsky-ineq+})). By the same way with the use of  Lemma \ref{b-lemma}B one can obtain the  following  local lower bound for the von Neumann entropy.
\smallskip

\begin{proposition}\label{S-LLB}
\emph{Let $\rho$ be an arbitrary state in $\S(\H)$. Then
\begin{equation}\label{S-LLB+}
S(\sigma)\geq \tilde{S}([\rho-\varepsilon I_{\H}]_+)-h^{\uparrow}(\varepsilon)
\end{equation}
for any state $\sigma$ in $\S(\H)$ such that  $\,\frac{1}{2}\|\rho-\sigma\|_1\leq \varepsilon\in(0,1]$, where $[\rho-\varepsilon I_{\H}]_+$ is the
positive part of the operator $\rho-\varepsilon I_{\H}$, $\tilde{S}(\cdot)$ is the extended von Neumann entropy defined in (\ref{S-ext}) and
$h^{\uparrow}(\varepsilon)$ is the function defined in (\ref{h+}).}\smallskip

\emph{The lower bounds (\ref{S-LLB+}) is faithful: its right hand side tends to $\,S(\rho)\leq+\infty\,$ as $\,\varepsilon\to0$.}
\end{proposition}\smallskip

\begin{remark}\label{S-LLB+r}
The lower bounds (\ref{S-LLB+}) is universal: it is applicable to any state $\rho$ (including states with infinite entropy).
Another advantage of lower bound  (\ref{S-LLB+}) consists in the fact that its r.h.s. is easily calculated, since
\begin{equation*}
S([\rho-\varepsilon I_{\H}]_+)=\sum_{i\in I_\varepsilon} \eta(\lambda_i^{\rho}-\varepsilon)-\eta\!\left(\sum_{i\in I_\varepsilon}(\lambda_i^{\rho}-\varepsilon)\right),\quad \eta(x)=-x\ln x,
\end{equation*}
where $\{\lambda^{\rho}_i\}$ is the  sequences of eigenvalues of $\rho$ and $I_\varepsilon$ is the \emph{finite} set of all $i$ such that $\lambda_i^{\rho}>\varepsilon$.
\end{remark}

Using Proposition \ref{S-LLB} one can improve the examples  considered after Proposition 1 in \cite{LLB} (obtained by the initial (non-optimised) version
of the method described in Section 3).

\subsection{Energy-type functionals}

Let $H$ be a positive operator on a separable Hilbert space $\H$.  In this section we analyze  the function
$$
E_{H}(\varrho)=\Tr H\varrho
$$
on the set $\S(\H)$ defined according to the rule (\ref{H-fun}). If the operator $H$ is bounded then the function $E_{H}$ is bounded and uniformly continuous on $\S(\H)$. The optimal continuity bound
for this function is given be the inequality
\begin{equation*}
|E_H(\rho)-E_H(\sigma)|\leq\frac{1}{2}\shs\|\rho-\sigma\|_1(\|H\|-E_0),
\end{equation*}
where $E_0$ is the infimum of the spectrum of $H$ \cite[Example 2]{QC}.

If the operator $H$ is unbounded then the function $E_{H}$
is not continuous on $\S(\H)$ and takes values in $[0,+\infty]$. The definition of $E_H$ implies that this function is affine and lower
semicontinuous on $\S(\H)$. Using the results of Section 3 one can obtain faithful optimal semicontinuity
bounds for $E_H$ assuming that the operator $H$ has representation (\ref{H-form}).\smallskip

\begin{proposition}\label{E-H-SCB} A) \emph{Let $H$ be a positive operator on $\H$ with representation (\ref{H-form}) and $\rho$ a state in $\S(\H)$ such that
$\Tr H^a \rho<+\infty$ for some $a\geq1$. Then $E_H(\rho)<+\infty$ and
\begin{equation}\label{E-H-SCB-1}
\!E_H(\rho)-E_H(\sigma)\leq \varepsilon^{1-1/a} \sqrt[a]{\Tr H^a (\hat{\rho}\wedge\varepsilon I_H)}=\varepsilon^{1-1/a} \sqrt[a]{\Tr H^a \rho-\Tr H^a [\hat{\rho}-\varepsilon I_H]_+}
\end{equation}
for any state $\sigma$ in $\S(\H)$ such that $\,\frac{1}{2}\|\rho-\sigma\|_1\leq \varepsilon$, where
\begin{equation*}
\hat{\rho}=\sum_{k=1}^{+\infty} \langle\tau_k|\rho|\tau_k\rangle|\tau_k\rangle\langle\tau_k|
\end{equation*}
is a state in $\S(\H)$ and $\hat{\rho}\wedge\varepsilon I_H$  and $\,[\hat{\rho}-\varepsilon I_{\H}]_+$ are the operators defined in (\ref{2-op}).}\smallskip

\noindent B) \emph{The semicontinuity
bound (\ref{E-H-SCB-1}) is faithful: the r.h.s. of (\ref{E-H-SCB-1}) tends to zero as $\varepsilon\to 0$.}\smallskip

\noindent C) \emph{If the minimal eigenvalue of $H$ is equal to zero then the semicontinuity
bound (\ref{E-H-SCB-1}) is optimal: for any $\varepsilon\in(0,1]$ there exist states  $\rho_{\varepsilon}$ and $\,\sigma_{\varepsilon}$
in $\,\S(\H)$ such that $\frac{1}{2}\|\rho-\sigma\|_1\leq \varepsilon$ and the equality holds in (\ref{E-H-SCB-1}).}

\end{proposition}\smallskip

\emph{Proof.} A)  The finiteness of $E_H(\rho)\doteq\Tr H\rho$ follows from the inequality
\begin{equation}\label{s-ineq}
\Tr H\rho\leq \sqrt[a]{\Tr H^a \rho}
\end{equation}
valid for any state $\rho$ in $\S(\H)$ (with possible value $+\infty$ in one or both sides). It is  clear that $\mathcal{D}(H^a)\subseteq\mathcal{D}(H)$. So, the last inequality is a direct corollary of the concavity of the function $\sqrt[a]{x}$, since for any state $\rho$ with finite $\Tr H^a\rho$ we have
$$
\Tr H\rho=\sum_{k=1}^{+\infty}  h_k p_k\quad \textrm{and} \quad \Tr H^a\rho=\sum_{k=1}^{+\infty}h^a_k p_k, \quad p_k=\langle\tau_k|\rho|\tau_k\rangle.
$$

Let $\sigma$ be any state in $\S(\H)$ such that $\textstyle\frac{1}{2}\|\rho-\sigma\|_1\leq\varepsilon$.
We may assume that  $\,E_H(\sigma)<+\infty$, since otherwise (\ref{E-H-SCB-1}) holds trivially.
Then  $\supp\sigma$ belongs to the subspace $\H_\mathcal{T}$ defined after (\ref{H-form}). Since the same
is true for the state $\rho$ by the condition $E_H(\rho)<+\infty$, the operators
\begin{equation*}
\hat{\rho}=\sum_{k=1}^{+\infty} \langle\tau_k|\rho|\tau_k\rangle|\tau_k\rangle\langle\tau_k|\quad \textrm{and}\quad
\hat{\sigma}=\sum_{k=1}^{+\infty} \langle\tau_k|\sigma|\tau_k\rangle|\tau_k\rangle\langle\tau_k|
\end{equation*}
are states in $\S(\H)$ such that $E_H(\hat{\rho})=E_H(\rho)$ and $E_H(\hat{\sigma})=E_H(\sigma)$.

Since $\hat{\rho}=\Pi(\rho)$ and $\hat{\sigma}=\Pi(\sigma)$, where $\Pi$ is the  channel $\,\varrho\mapsto \sum_{k=1}^{+\infty} \langle\tau_k|\varrho|\tau_k\rangle|\tau_k\rangle\langle\tau_k|\,$ from $\T(\H_\mathcal{T})$ to itself, we have
$$
\epsilon\doteq\frac{1}{2}\|\hat{\rho}-\hat{\sigma}\|_1\leq\frac{1}{2}\|\rho-\sigma\|_1\leq \varepsilon.
$$
So, since the states $\hat{\rho}$ and $\hat{\sigma}$ commute and $E_H$ is an affine nonnegative function on $\S(\H)$, by
applying the arguments from the proof of Lemma \ref{b-lemma}A in Section 3
with $\S_0=\S(\H)$ to the function $f=E_H$  (the inequalities (\ref{imp-ineq})-(\ref{imp-ineq+})) we obtain
$$
E_H(\hat{\rho})-E_H(\hat{\sigma})\leq \varepsilon\Tr H\tau_+\leq\varepsilon\sqrt[a]{\Tr H^a\tau_+}\leq\varepsilon^{1-1/a}\sqrt[a]{\Tr H^a (\hat{\rho}\wedge\varepsilon I_H)},
$$
where $\tau_+=(1/\epsilon)[\hat{\rho}-\hat{\sigma}]_+$ and the second inequality follows  from the inequality (\ref{s-ineq})
applied to the state $\tau_+$. To complete the proof of (\ref{E-H-SCB-1}) it suffices to note that $E_H(\hat{\rho})=E_H(\rho)$, $E_H(\hat{\sigma})=E_H(\sigma)$ and $\Tr H^a \rho= \Tr H^a\hat{\rho}=\Tr H^a(\hat{\rho}\wedge\varepsilon I_H)+\Tr H^a [\hat{\rho}-\varepsilon I_H]_+$.\smallskip

Claim B) is obvious. To prove claim C) consider the states $\rho_{\varepsilon}=|\tau_k\rangle\langle\tau_k|$, $k>1$, and $\,\sigma_{\varepsilon}=\varepsilon|\tau_1\rangle\langle\tau_1|+(1-\varepsilon)\rho_{\varepsilon}$. Then  $\frac{1}{2}\|\rho_{\varepsilon}-\sigma_{\varepsilon}\|_1=\varepsilon$
and the assumption $h_1=0$ implies
$$
E_H(\rho_{\varepsilon})-E_H(\sigma_{\varepsilon})=h_k\varepsilon=\varepsilon^{1-1/a} \sqrt[a]{\Tr H^a (\rho_{\varepsilon}\wedge\varepsilon I_H)}=\varepsilon^{1-1/a} \sqrt[a]{\Tr H^a (\hat{\rho}_{\varepsilon}\wedge\varepsilon I_H)}.\;\Box
$$
\smallskip

\begin{remark}\label{E-H-SCB-r} For any $a>1$ the semicontinuity
bound (\ref{E-H-SCB-1})  can be  replaced by the simple semicontinuity
bound
\begin{equation*}
E_H(\rho)-E_H(\sigma)\leq\varepsilon^{1-1/a} \sqrt[a]{\Tr H^a \rho}
\end{equation*}
valid for any states $\rho$ and $\sigma$ in $\S(\H)$ such that $\frac{1}{2}\|\rho-\sigma\|_1\leq \varepsilon$ and $\Tr H^a \rho<+\infty$.

The advantage the semicontinuity
bound (\ref{E-H-SCB-1}) with $a\gg 1$ (in comparison with the same bound with $a$ close to $1$) is a higher rate of convergence to zero of its  right-hand side.
The price for this advantage is the requirement of finiteness of $\Tr H^a \rho$ (which is stronger the greater the parameter $a$).
$\Box$\medskip
\end{remark}

It is well known that the function $E_H$ is continuous on the compact\footnote{The compactness of the set $\C_{H^a\!\!,E}$ follows from the Lemma in \cite{H-c-w-c}.} set
$$
\C_{H^a\!\!,E}\doteq\left\{\rho\in\S(\H)\,|\,\Tr H^a\rho\leq E\right\},\quad E>0,
$$
for any $a>1$. Proposition
\ref{E-H-SCB} and Remark \ref{E-H-SCB-r} imply the simple continuity bound for $E_H$ on $\C_{H^a\!\!,E}$ presented in the following\smallskip

\begin{corollary}\label{E-H-CB} \emph{Let $H$ be a positive operator on $\H$ with representation (\ref{H-form}) and $a>1$ be arbitrary. Then
\begin{equation}\label{E-H-CB+}
|E_H(\rho)-E_H(\sigma)|\leq \varepsilon^{1-1/a} \sqrt[a]{E}
\end{equation}
for any states $\rho$ and $\sigma$ in $\C_{H^a\!\!,E}$ such that $\,\frac{1}{2}\|\rho-\sigma\|_1\leq \varepsilon$.}
\end{corollary}\smallskip

The non-convergence to zero of the right-hand side of (\ref{E-H-CB+}) with $a\leq 1$ is not surprising, since it is easy to show that the function $E_H$ is not continuous on the  set $\C_{H^a\!\!,E}$ for any $a\leq 1$.

\subsection{Quantum relative entropy}

Continuity bounds for the quantum relative entropy (defined in (\ref{URE-def})) in the finite dimensional settings are
obtained in \cite{RE-CB-1,RE-CB-2} by applying the special modification of the Alicki-Fannes-Winter method proposed therein.

In this section, we use the results of the previous sections to obtain continuity and semicontinuity bounds for the
function $\varrho\mapsto D(\varrho\shs\|\omega)$ on the set of states of an infinite-dimensional quantum system,
where $\omega$ is a given state of this system.

If $\omega$ is a finite rank state then the
function $\varrho\mapsto D(\varrho\shs\|\omega)$ is uniformly continuous on the set of states $\varrho$ such that
$\supp \varrho\subseteq\supp\omega$. Continuity bound for this function (depending on the minimal eigenvalue of $\omega$) is presented in Corollary
5.9 in \cite{RE-CB-1}.

If $\omega$ is an infinite rank state then the
function $\varrho\mapsto D(\varrho\shs\|\omega)$ is not continuous on the set of states $\varrho$ such that
$\supp \varrho\subseteq\supp\omega$ and it takes the value $+\infty$ at some states of this set. Moreover,
this function is not continuous\footnote{This can be shown by considering the sequence of states $\varrho_n=(1-1/h_k)|\phi_1\rangle\langle\phi_1|+(1/h_k)|\phi_k\rangle\langle\phi_k|$, where $h_k=-\ln \lambda^{\omega}_k$, provided that $\omega=\sum_{k=1}^{+\infty}\lambda^{\omega}_k |\phi_k\rangle\langle\phi_k|$ is a spectral representation of
$\omega$.}  on the set
$$
\{\varrho\in\S(\H)\,|\,\Tr\varrho\shs(-\ln\omega)\leq E\}
$$
for any $E>0$ on which $D(\varrho\shs\|\omega)\leq E$ (this follows from the
representation (\ref{re-exp})). \smallskip

However, the
function $\varrho\mapsto D(\varrho\|\omega)$ is uniformly continuous on the set
$$
\{\varrho\in\S(\H)\,|\,\Tr\varrho\shs(-\ln\omega)^a\leq E\}
$$
for any $a>1$ and $E>0$. A simple continuity bound for this function  is given by the following\smallskip

\begin{proposition}\label{RE-CB}
\emph{Let $\omega$ be a faithful  state in $\S(\H)$ and
$H=c(-\ln\omega+\ln\lambda^{\omega}_1I_{\H})\,$ be a positive densely defined operator on $\H$ satisfying condition (\ref{star}), where $c>0$ and $\lambda^{\omega}_1$ is the maximal eigenvalue of $\,\omega$. Let $a>1$ be arbitrary.  Then
\begin{equation}\label{RE-CB+}
|D(\rho\shs\|\shs\omega)-D(\sigma\|\shs\omega)|\leq (1/c)\varepsilon^{1-1/a} \sqrt[a]{E}+\varepsilon F_{H}(\sqrt[a]{E/\varepsilon}) +h^{\uparrow}(\varepsilon)
\end{equation}
for any states $\rho$ and $\sigma$ in $\S(\H)$ such that $\frac{1}{2}\|\rho-\sigma\|_1\leq \varepsilon$ and
$\Tr H^a\rho,\Tr H^a\sigma \leq E$, where $F_{H}$ is the function defined in (\ref{F-def}) and
$h^{\uparrow}$ is the function defined in (\ref{h+}).}
\end{proposition}\smallskip

\emph{Proof.} Note first that that the positive operator $H^a$ satisfies the Gibbs condition (\ref{H-cond}).
Indeed, the inequality (\ref{s-ineq}) and the definition (\ref{F-def}) imply that
\begin{equation}\label{F-ineq}
F_{H^a}(E)\leq F_{H}(\sqrt[a]{E})\qquad \forall E>0.
\end{equation}
By Proposition 1 in \cite{EC} there exist positive numbers $p$ and $q$ such that
$F_{H}(E)\leq pE+q$ for all $E>0$. Thus, $F_{H^a}(E)=o(E)$ as $E\to+\infty$ and the required claim follows from the  equivalence of (\ref{H-cond}) and (\ref{H-cond-a}).

Since the operator $H^a$ satisfies the  condition (\ref{H-cond}), Theorem \ref{SCB-m}B and Remark \ref{SCB-r} in Section 4.1  imply (due to the assumptions
 $\frac{1}{2}\|\rho-\sigma\|_1\leq \varepsilon$ and
$\Tr H^a\rho,\Tr H^a\sigma \leq E$)
the finiteness of $S(\rho)$ and $S(\sigma)$ and the inequality
\begin{equation}\label{S-CB}
|S(\rho)-S(\sigma)|\leq\varepsilon F_{H^a}(E/\varepsilon)+h^{\uparrow}(\varepsilon)\leq\varepsilon F_{H}(\sqrt[a]{E/\varepsilon})+h^{\uparrow}(\varepsilon),
\end{equation}
where the second inequality follows from (\ref{F-ineq}).

Since representation (\ref{re-exp}) implies that $D(\varrho\shs\|\omega)=(1/c)\Tr H\varrho-\ln\lambda^{\omega}_1-S(\varrho)$, $\varrho=\rho,\sigma$,  inequality (\ref{RE-CB+}) follows from
(\ref{S-CB}) and  Corollary \ref{E-H-CB}. $\Box$\smallskip

\begin{corollary}\label{RE-CB-c} \emph{Let $H$ be a positive densely defined operator satisfying condition (\ref{star}) and $\beta>0$ such that $\Tr e^{-\beta H}<+\infty$. Let $\,\omega=e^{-\beta H}/\Tr e^{-\beta H}$ and $a>1$ be arbitrary. Then
\begin{equation}\label{RE-CB-c+}
|D(\rho\shs\|\shs\omega)-D(\sigma\|\shs\omega)|\leq \beta\varepsilon^{1-1/a} \sqrt[a]{E}+\varepsilon F_{H}(\sqrt[a]{E/\varepsilon}) +h^{\uparrow}(\varepsilon)
\end{equation}
for any states $\rho$ and $\sigma$ in $\S(\H)$ such that $\frac{1}{2}\|\rho-\sigma\|_1\leq \varepsilon$ and
$\Tr H^a\rho,\Tr H^a\sigma \leq E$, where  $F_{H}$ is the function defined in (\ref{F-def}) and $\,h^{\uparrow}$ is the function defined in (\ref{h+}).}
\end{corollary}\smallskip

\begin{remark}\label{RE-CB-r} The continuity bounds (\ref{RE-CB+}) and (\ref{RE-CB-c+}) can be improved by replacing the term $\varepsilon F_{H}(\sqrt[a]{E/\varepsilon}) +h^{\uparrow}(\varepsilon)$ with the term $ \varepsilon F_{H_1}(\sqrt[a]{E/\varepsilon})+h(\varepsilon)$ for  all  $\varepsilon\leq 1-1/Z_{H^a}\!(E)$, where $H_1$ is the operator defined in (\ref{H-rep+}) with $m=1$ and $Z_{H^a}$ is the function defined in (\ref{Z}) with $H$ replaced by $H^a$. This follows from the proof of Proposition \ref{RE-CB}. If the parameter $Z_{H^a}(E)$ is difficult to determine  then one can use Remark 4 in \cite{FCB} with $m=1$ to show that for any $\varepsilon\in(0,1]$ the term $\varepsilon F_{H}(\sqrt[a]{E/\varepsilon}) +h^{\uparrow}(\varepsilon)$ can be replaced by the expression
$$
\max_{x\in[0,\shs\min\{\varepsilon,E/h^a_{2}\}]}\left\{xF_{H_1}(\sqrt[a]{E/x})+h(x)\right\},
$$
where it is assumed that the operator $H$ has representation (\ref{H-form}) (one can prove that this maximum is equal to  $ \varepsilon F_{H_1}(\sqrt[a]{E/\varepsilon})+h(\varepsilon)$ for all sufficiently small $\varepsilon$).

\end{remark}\smallskip

\begin{example}\label{1-mode} Assume that $\H$ is a Hilbert space describing a one-mode quantum oscillator (see \cite[Ch.12]{H-SCI}) and  $H=\hat{N}\doteq a^{\dag}a\,$ is  the number operator on $\H$ with the representation
\begin{equation}\label{N-m}
\hat{N}=\sum_{k=1}^{+\infty}(k-1)|\tau_k\rangle\langle \tau_k|,
\end{equation}
where $\{\tau_k\}_{k=1}^{+\infty}$ is the Fock basis in $\H$. It is easy to see that the operator $\hat{N}$ satisfies conditions (\ref{H-cond}) and (\ref{star}) and that
$\beta_{\hat{N}}(E)=\ln(1+1/E)$, $F_{\hat{N}}(E)=g(E)$ for any $E>0$ and $F_{\hat{N}_1}(E)=Eh(1/E)$ for any $E>1$, where $g$ is the function defined in (\ref{g-def}), $h$ is the binary entropy and $\hat{N}_1$ is the operator defined in (\ref{H-rep+}) with $H=\hat{N}$ and $m=1$ ($\hat{N}_1$ is the operator with the representation (\ref{N-m}) in which the summation is over all $k\in\N\cup[2,+\infty)$). For given $E>0$ let
\begin{equation*}
\gamma_{\hat{N}}(E)\doteq\frac{e^{-\beta_{\hat{N}}(E)\hat{N}}}{\Tr e^{-\beta_{\hat{N}}(E)\hat{N}}}=\frac{1}{E+1}\sum_{k=1}^{+\infty}\left[\frac{E}{E+1}\right]^{k-1}|\tau_k\rangle\langle \tau_k|
\end{equation*}
be the Gibbs state corresponding to the number of quanta $E$.

By Corollary \ref{RE-CB-c}  for a given arbitrary $a>1$ we have
\begin{equation}\label{RE-CB-c++}
|D(\rho\shs\|\shs\gamma_{\hat{N}}(E))-D(\sigma\|\shs\gamma_{\hat{N}}(E))|\leq \varepsilon^{1-1/a} \sqrt[a]{E}\ln(1+1/E)+\varepsilon g(\sqrt[a]{E/\varepsilon}) +h^{\uparrow}(\varepsilon)
\end{equation}
for any states $\rho$ and $\sigma$ in $\S(\H)$ such that $\frac{1}{2}\|\rho-\sigma\|_1\leq \varepsilon$ and
$\Tr \hat{N}^a\rho,\Tr \hat{N}^a\sigma \leq E$.

According to Remark \ref{RE-CB-r} the continuity bound (\ref{RE-CB-c++}) can be improved by replacing the term $\varepsilon g(\sqrt[a]{E/\varepsilon})+h^{\uparrow}(\varepsilon)$ with the term $\varepsilon^{1-1/a}\sqrt[a]{E} h(\sqrt[a]{\varepsilon/E})+h(\varepsilon)$ for all $\varepsilon\leq 1-1/Z_{\hat{N}^a}(E)$. As a result, we obtain the inequality
\begin{equation*}
|D(\rho\shs\|\shs\gamma_{\hat{N}}(E))-D(\sigma\|\shs\gamma_{\hat{N}}(E))|\leq \varepsilon^{1-1/a} \sqrt[a]{E}\left(\ln(1+1/E)+h(\sqrt[a]{\varepsilon/E})\right)+h(\varepsilon)
\end{equation*}
valid for any states $\rho$ and $\sigma$ in $\S(\H)$ such that $\frac{1}{2}\|\rho-\sigma\|_1\leq \varepsilon\leq 1-1/Z_{\hat{N}^a}(E)$ and
$\Tr \hat{N}^a\rho,\Tr \hat{N}^a\sigma \leq E$. Because of the absence of a simple expression for $\beta_{\hat{N}^a}(E)$ the parameter $Z_{\hat{N}^a}(E)$ is difficult to determine
(although it can be  found  numerically). Another way is to replace the term $\varepsilon g(\sqrt[a]{E/\varepsilon})+h^{\uparrow}(\varepsilon)$ in (\ref{RE-CB-c++})
by the expression
$$
\max_{x\in[0,\shs\min\{\varepsilon,E\}]}\left\{x^{1-1/a}\sqrt[a]{E} h(\sqrt[a]{x/E})+h(x)\right\}
$$
for any $\varepsilon\in(0,1]$. This expression coincides with $\,\varepsilon^{1-1/a}\sqrt[a]{E} h(\sqrt[a]{\varepsilon/E})+h(\varepsilon)\,$ for all sufficiently small $\varepsilon$, but to find the region of coincidence  we have to solve a  transcendental equation. $\Box$
\end{example}\medskip

It is known that the function $\varrho\mapsto D(\varrho\shs\|\omega)$ is continuous on the set
of all states $\varrho$ such that $c\varrho\leq\omega$, where $c\in(0,1)$ and $\,\omega$ is an arbitrary state \cite[Example 1]{REC}.
To obtain a simple continuity bound for this function valid for commuting states one can use the following\smallskip

\begin{proposition}\label{RE-SCB}
\emph{Let $\,\omega$ be an arbitrary state in $\,\S(\H)$ and $\,c\in(0,1)$. Then
\begin{equation}\label{RE-SCB+}
D(\rho\shs\|\shs\omega)-D(\sigma\|\shs\omega)\leq (1/c)\eta^{\uparrow}(c\varepsilon)+h^{\uparrow}(\varepsilon)
\end{equation}
for any states $\rho$ and $\sigma$ in $\S(\H)$ such that $\frac{1}{2}\|\rho-\sigma\|_1\leq \varepsilon$,
$c\rho\leq\omega$ and either $[\rho,\sigma]=0$ or $[\rho,\omega]=0$, where $h^{\uparrow}$ is the function defined in (\ref{h+}) and}
\begin{equation}\label{eta+}
\eta^{\uparrow}(x)=\left\{\begin{array}{l}
        \!\!-x\ln x\;\; \textrm{if}\;\;  x\in(0,1/e]\\
        \;\;1/e\;\quad \textrm{if}\;\;  x\in(1/e,1]
        \end{array}\right.\!\!\!,\qquad \eta^{\uparrow}(0)=0.
\end{equation}

\emph{If $\frac{1}{2}\|\rho-\sigma\|_1=\varepsilon$ and $\,[\rho,\sigma]=0$  then (\ref{RE-SCB+}) holds with $\eta^{\uparrow}(c\varepsilon)$ and $h^{\uparrow}(\varepsilon)$
replaced, respectively, by $\eta(c\varepsilon)$ and $h(\varepsilon)$, where $\eta(x)=-x\ln x$ and $h$ is the binary entropy.}
\end{proposition}\smallskip

\emph{Proof.} We will use Lindblad's extension of the quantum relative entropy to positive operators $\rho$ and
$\sigma$ in $\mathfrak{T}(\mathcal{H})$  defined as
\begin{equation*}
D(\rho\shs\|\shs\sigma)=\sum_i\langle\varphi_i|\,\rho\ln\rho-\rho\ln\sigma+\sigma-\rho\,|\varphi_i\rangle,
\end{equation*}
where $\{\varphi_i\}$ is the orthonormal basis of
eigenvectors of the operator  $\rho$ and it is assumed that $\,D(0\|\shs\sigma)=\Tr\sigma\,$ and
$\,D(\rho\shs\|\sigma)=+\infty\,$ if $\,\mathrm{supp}\rho\shs$ is not
contained in $\shs\mathrm{supp}\shs\sigma$ (in particular, if $\rho\neq0$ and $\sigma=0$)
\cite{L-2}. If the extended von Neumann entropy $\tilde{S}(\rho)$ of $\rho$ (defined in (\ref{S-ext})) is finite
then
\begin{equation}\label{re-exp+}
D(\rho\shs\|\shs\sigma)=\Tr\rho(-\ln\sigma)-\tilde{S}(\rho)-\eta(\Tr\rho)+\Tr\sigma-\Tr\rho,
\end{equation}
where $\Tr\rho(-\ln\sigma)$ is defined according to the rule (\ref{H-fun}) and $\eta(x)=-x\ln x$. \smallskip

The function $(\rho,\sigma)\mapsto D(\rho\shs\|\shs\sigma)$ is nonnegative lower semicontinuous and jointly convex on
$\T_+(\H)\times\T_+(\H)$. We will exploit the following properties of this function:
\begin{itemize}
  \item for any $\rho,\sigma\in\T_+(\H)$ and $c\geq0$ the following equalities  hold (with possible values $+\infty$ in both sides):\footnote{Here and in what follows we assume
  that $0\times D(\rho\shs\|\shs \sigma)=0$ in any cases (including the case $D(\rho\shs\|\shs \sigma)=+\infty$) and that $0\ln0=0$.}
  \begin{equation}\label{D-mul}
  D(c\rho\shs\|\shs c\sigma)=cD(\rho\shs\|\shs \sigma),\qquad\qquad\qquad\qquad\quad\;
  \end{equation}
  \begin{equation}\label{D-c-id}
  D(\rho\shs\|\shs c\sigma)=D(\rho\shs\|\shs\sigma)-\Tr\rho\ln c+(c-1)\Tr\sigma;
  \end{equation}
  \item for any $\rho,\sigma$ and $\omega$ in $\T_+(\H)$ the following inequality holds (with possible values $+\infty$ in one or both sides)
  \begin{equation}\label{re-ineq}
  D(\rho\shs\|\shs\sigma+\omega)\leq D(\rho\shs\|\shs\sigma)+\Tr\omega.
  \end{equation}
 \end{itemize}

If the extended von Neumann entropy of the operators $\rho$, $\sigma$  and $\omega$ is finite then inequality (\ref{re-ineq}) is easily proved by using representation (\ref{re-exp+}) and the operator monotonicity of the logarithm. In the general case, this inequality can be proved by approximation using Lemma 4 in \cite{L-2}.

Note that the condition $c\rho\leq\omega$ implies $D(\rho\shs\|\shs\omega)<+\infty$  (due to (\ref{D-c-id}) and (\ref{re-ineq})).

Consider first the case $[\rho,\sigma]=0$. Since the r.h.s. of (\ref{RE-SCB+}) is a non-decreasing function of $\varepsilon$, we may assume that $\frac{1}{2}\|\rho-\sigma\|_1=\varepsilon$.   The function
$\varrho\mapsto D(\varrho\shs\|\shs\omega)$ is convex and satisfies inequality (\ref{LAA-1})
with $a_f=h$ (by Proposition 5.24 in \cite{O&P}). So, since the states $\rho$ and $\sigma$ commute, by
applying the arguments from the proof of Lemma \ref{g-ob} in Section 3
with $\S_0=\S(\H)$ to this function  (the relations in (\ref{omega-star}) and below) one can  show that
\begin{equation}\label{two-ineq}
\varepsilon\tau_+\leq\rho\qquad
\textrm{and}\qquad D(\rho\shs\|\shs\omega)-D(\sigma\shs\|\shs\omega)\leq \varepsilon D(\tau_+\shs\|\shs\omega)+h(\varepsilon),
\end{equation}
where $\tau_+=(1/\varepsilon)[\rho-\sigma]_+$. By applying (\ref{D-mul}), (\ref{D-c-id}) and (\ref{re-ineq}) we obtain
$$
\begin{array}{rl}
\varepsilon D(\tau_+\shs\|\shs\omega)\, = & \!\!(1/c)D(c\varepsilon\tau_+\shs\|\shs c\varepsilon\omega)\\\\= &\!\!
(1/c)(D(c\varepsilon\tau_+\shs\|\shs\omega)-c\varepsilon\ln(c\varepsilon)+(c\varepsilon-1))\\\\ \leq &\!\!
(1/c)(D(c\varepsilon\tau_+\shs\|\shs c\varepsilon\tau_+)+\Tr(\omega-c\varepsilon\tau_+)+\eta( c\varepsilon)+(c\varepsilon-1))\\\\= &\!\!
(1/c)(1-c\varepsilon+\eta(c\varepsilon)+(c\varepsilon-1))=(1/c)\eta(c\varepsilon)\leq(1/c)\eta^{\uparrow}(c\varepsilon),
\end{array}
$$
where we have used that $\omega-c\varepsilon\tau_+$ is a positive operator (due to the assumption $c\rho\leq\omega$ and the first inequality in (\ref{two-ineq})).
Thus, the second inequality in (\ref{two-ineq}) implies (\ref{RE-SCB+}).
\smallskip

Consider the case $[\rho,\omega]=0$. We may assume that  $\,\mathrm{supp}\sigma\subseteq\mathrm{supp}\shs\omega$, since otherwise
$D(\sigma\shs\|\shs\omega)=+\infty$ and (\ref{RE-SCB+}) holds trivially. So, as the  condition $c\rho\leq\omega$
implies $\,\mathrm{supp}\rho\subseteq\mathrm{supp}\shs\omega$, we may assume that $\omega$ is a faithful state.
Let $\{\phi_k\}_{k=1}^{+\infty}$ be a  basis in $\H$ in which both states  $\rho$ and $\omega$ are diagonisable and
$\Pi$ be the  channel $\,\varrho\mapsto \sum_{k=1}^{+\infty} \langle\phi_k|\varrho|\phi_k\rangle|\phi_k\rangle\langle\phi_k|\,$ from $\T(\H)$ to itself.
Then, $\Pi(\rho)=\rho$, $\Pi(\omega)=\omega$ and hence
$$
\textstyle\frac{1}{2}\|\rho-\Pi(\sigma)\|_1\leq\frac{1}{2}\|\rho-\sigma\|_1\leq \varepsilon.
$$
By monotonicity of the relative entropy we have
$$
D(\rho\shs\|\shs\omega)-D(\sigma\|\shs\omega)\leq D(\rho\shs\|\shs\omega)-D(\Pi(\sigma)\|\shs\omega)\leq (1/c)\eta^{\uparrow}(c\varepsilon)+h^{\uparrow}(\varepsilon),
$$
where the second inequality follows from the proved part of the proposition, since
the states $\rho$ and $\Pi(\sigma)$ commute.
\smallskip

The last claim of the proposition follows from the above proof. $\Box$
\smallskip

\begin{corollary}\label{RE-SCB-c}
\emph{Let $\omega$ be an arbitrary state in $\S(\H)$ and $c\in(0,1)$. Then
\begin{equation}\label{RE-SCB+c}
|D(\rho\|\shs\omega)-D(\sigma\|\shs\omega)|\leq (1/c)\eta^{\uparrow}(c\varepsilon)+h^{\uparrow}(\varepsilon)
\end{equation}
for any states $\rho$ and $\sigma$ in $\S(\H)$ such that $\frac{1}{2}\|\rho-\sigma\|_1\leq \varepsilon$,
$c\rho,c\sigma\leq\omega$ and $[\rho,\sigma]=0$, where $\eta^{\uparrow}$ and $h^{\uparrow}$ are the functions defined, respectively, in (\ref{eta+}) and (\ref{h+}).}\smallskip

\emph{If $\,\frac{1}{2}\|\rho-\sigma\|_1=\varepsilon$  then (\ref{RE-SCB+c}) holds with $\eta^{\uparrow}(c\varepsilon)$ and $h^{\uparrow}(\varepsilon)$
replaced, respectively, by $\eta(c\varepsilon)$ and $h(\varepsilon)$, where $\eta(x)=-x\ln x$ and $h$ is the binary entropy.}
\end{corollary}\smallskip

\begin{remark}\label{RE-SCB-r} The (semi)continuity bounds (\ref{RE-SCB+}) and (\ref{RE-SCB+c}) can be improved by replacing their right hand sides by
the expression $$
\{(1/c)\eta(c\varepsilon)+h(\varepsilon)\}^{\uparrow}_{\varepsilon}\doteq \sup_{\delta\in(0,\varepsilon]}[(1/c)\eta(c\delta)+h(\delta)].
$$
This follows from the proof of Proposition \ref{RE-SCB}.
\end{remark}\medskip

\subsection{Quantum conditional entropy}

\subsubsection{Continuity bounds for commuting states}

It is mentioned in Section 3.2 that the (extended) quantum conditional entropy $S(A\vert B)$ (defined by the formula in (\ref{ce-ext})
on the set $\S_1(\H_{AB})\doteq\left\{\rho\in\mathfrak{S}(\mathcal{H}_{AB})\,\vert\,S(\rho_{A})<+\infty\shs\right\}$)  belongs to the class $L_2^1(2,1)\doteq L_2^1(2,1\vert\,\S_1(\H_{AB}))$ (in the settings $A_1=A, A_2=B$). Thus, Corollaries \ref{main-1-c}A
and \ref{main-c}A in Section 3.2 imply the following
\smallskip\pagebreak

\begin{proposition}\label{QCE-GC}  \emph{Let $AB$ be an infinite-dimensional bipartite quantum system.}

\noindent A)\emph{ The inequality
\begin{equation}\label{QCE-GC+}
    \vert S(A\vert B)_{\rho}-S(A\vert B)_{\sigma}\vert \leq \{2\varepsilon \ln d+h(\varepsilon)\}_{\varepsilon}^{\uparrow}\leq2\varepsilon \ln d+h^{\uparrow}(\varepsilon)
\end{equation}
holds for any commuting states $\rho$ and  $\sigma$ in $\,\S(\H_{AB})$ such that $\,\rank\rho_{A}, \rank\sigma_A\leq d<+\infty$ and
$\frac{1}{2}\|\rho-\sigma\|_1\leq\varepsilon$, where $\,\{...\}^{\uparrow}_{\varepsilon}$ is defined in (\ref{men-def+}) and $\,h^{\uparrow}$ is the function defined in (\ref{h+}).}\smallskip

\noindent B)  \emph{Let $H$ be a positive operator on the Hilbert space $\H_A$ satisfying conditions (\ref{H-cond})
and (\ref{star}). Then  the inequality
\begin{equation}\label{QCE-GC++}
    \vert S(A\vert B)_{\rho}-S(A\vert B)_{\sigma}\vert \leq 2\varepsilon F_{H}(E/\varepsilon)+h^{\uparrow}(\varepsilon)
\end{equation}
holds for any commuting states $\rho$ and  $\sigma$ in $\,\S(\H_{AB})$ such that $\,\Tr H\rho_{A},\,\Tr H\sigma_{A}\leq E\,$
and $\frac{1}{2}\|\rho-\sigma\|_1\leq\varepsilon$, where $F_{H}$ and $\,h^{\uparrow}$ are the functions defined, respectively,  in (\ref{F-def}) and  (\ref{h+}).}
\end{proposition}

\smallskip

Continuity bound (\ref{QCE-GC+}) should be compared with Winter's continuity bound (cf.\cite{W-CB})
\begin{equation}\label{CE-W-CB}
    \vert S(A\vert B)_{\rho}-S(A\vert B)_{\sigma}\vert \leq 2\varepsilon \ln d+g(\varepsilon)
\end{equation}
valid for \emph{arbitrary} states $\rho$ and  $\sigma$ in $\,\S(\H_{AB})$ such that $\,\supp\rho_{A}, \supp\sigma_A\subseteq \H_d$ and
$\frac{1}{2}\|\rho-\sigma\|_1\leq\varepsilon$, where $\H_d$ is any $d$-dimensional subspace of $\H_A$ and $g$ is the function defined in (\ref{g-def}).\smallskip

It is clear that the r.h.s. of (\ref{QCE-GC+})  is strictly less than the r.h.s. of  (\ref{CE-W-CB}) for any nonzero $\varepsilon$. Note also
that the condition $\,\supp\rho_{A}, \supp\sigma_A\subseteq \H_d$ is stronger than the condition $\,\rank\rho_{A}, \rank\sigma_A\leq d$.

Continuity bound (\ref{QCE-GC++}) demonstrates a more significant advantage (both in accuracy and mathematical simplicity) compared to the universal continuity bounds
for the quantum conditional entropy with energy-type constraint presented in \cite[Meta-Lemma 17]{W-CB} and \cite[Proposition 5]{QC}. Of course, this advantage is due to the fact that we restrict attention to commuting states.

\subsubsection{Semicontinuity bounds for quantum-classical states}

In this subsection we use the advanced version of the AFW-method in the quasi-classical settings to improve the semicontinuity bounds and local lower bounds for the  quantum conditional entropy (QCE) on the set of
quantum-classical states of a bipartite system obtained by the original version of this method in \cite{LCB,LLB}.

States $\rho$ and $\sigma$ of bipartite quantum system $AB$ are called quantum-classical (q-c) if they have the form
\begin{equation}\label{qc-states}
\rho=\sum_{k} p_k\, \rho_k\otimes \vert k\rangle\langle k\vert \quad \textrm{and}\quad \sigma=\sum_{k} q_k\, \sigma_k\otimes \vert k\rangle\langle k\vert,
\end{equation}
where $\{p_k,\rho_k\}$ and $\{q_k,\sigma_k\}$ are ensembles of states in $\S(\H_A)$ and $\{\vert k\rangle\}$ a fixed orthonormal basis in $\H_B$.
By using definition (\ref{ce-ext}) one can show (see the proof of Corollary 3 in \cite{Wilde-CB}) that
\begin{equation}\label{ce-rep}
S(A\vert B)_{\rho}=\sum_kp_kS(\rho_k)\quad \textrm{and} \quad S(A\vert B)_{\sigma}=\sum_kq_kS(\sigma_k)
\end{equation}
provided that $\,S(\rho_A), S(\sigma_A)<+\infty$.
The expressions in (\ref{ce-rep}) allow us to define the QCE on the set $\S_{\mathrm{qc}}(\H_{AB})$ of all q-c states
having the form (\ref{qc-states}) (including the q-c states $\rho$ with $\,S(\rho_A)=+\infty\,$ for which definition (\ref{ce-ext}) does not work).
It is easy to show that the QCE defined in such a way on the convex set $\S_{\mathrm{qc}}(\H_{AB})$ satisfies
inequalities (\ref{LAA-1}) and (\ref{LAA-2}) with $a_f=0$ and $b_f=h$ with possible values $+\infty$ in both sides.
Using this and the  upper bound (\ref{ce-ub}) we see that $S(A\vert B)$ is a nonnegative function on
the convex set $\S_{\mathrm{qc}}(\H_{AB})$  belonging to the class $L^1_2(1,1\vert \S_{\mathrm{qc}}(\H_{AB}))$ (in terms of Section 3.2).
So, by repeating the arguments from the proof of Proposition 3 in \cite{LCB}  (based on the Mirsky inequality (\ref{Mirsky-ineq+})) and the  upper bound (\ref{ce-ub}) one can prove the following improved version of this proposition.\smallskip

\begin{proposition}\label{qce-qc-1} \emph{Let $AB$ be an infinite-dimensional bipartite quantum system.}\smallskip

\noindent A) \emph{If $\rho$ is a state in $\S_{\mathrm{qc}}(\H_{AB})$ such that $\,\rank\rho_A$ is finite then
\begin{equation}\label{qce-qc-1+}
   S(A\vert B)_{\rho}-S(A\vert B)_{\sigma}\leq\{\varepsilon\ln(\rank\rho_A)+h(\varepsilon)\}_{\varepsilon}^{\uparrow}\leq\varepsilon\ln(\rank\rho_A)+h^{\uparrow}(\varepsilon)
\end{equation}
for any state $\sigma$ in $\S_{\mathrm{qc}}(\H_{AB})$ such that $\,\frac{1}{2}\|\rho-\sigma\|_1\leq\varepsilon$, where $\,\{...\}^{\uparrow}_{\varepsilon}$ is defined in (\ref{men-def+}) and $\,h^{\uparrow}$ is the function defined in (\ref{h+}).}
\smallskip

\noindent B) \emph{Let $H$ be a positive operator on $\H_A$  satisfying conditions (\ref{H-cond}) and (\ref{star}). If $\rho$ is a  state in $\,\S_{\mathrm{qc}}(\H_{AB})$  such that $E\doteq\Tr H\rho_A<+\infty$ then
\begin{equation}\label{qce-qc-2}
   S(A\vert B)_{\rho}-S(A\vert B)_{\sigma}\leq \varepsilon F_H((E-E_{H,\shs\varepsilon}(\rho))/\varepsilon)+h^{\uparrow}(\varepsilon)\leq \varepsilon F_H(E/\varepsilon)+h^{\uparrow}(\varepsilon)
\end{equation}
for any  state $\sigma$ in $\S_{\mathrm{qc}}(\H_{AB})$ such that $\,\frac{1}{2}\|\rho-\sigma\|_1\leq\varepsilon$, where
\begin{equation*}
E_{H,\shs\varepsilon}(\rho)\doteq \sum_{k}\Tr H[p_k\rho_k-\varepsilon I_{A}]_+,
\end{equation*}
$[p_k\rho_k-\varepsilon  I_{A}]_+$ is the positive part of the Hermitian operator $\,p_k\rho_k-\varepsilon  I_{A}$ and  $h^{\uparrow}$ is the function defined in (\ref{h+}).}\smallskip

\emph{Inequality (\ref{qce-qc-1+}) and both inequalities in (\ref{qce-qc-2}) are faithful semicontinuity  bounds for the QCE on $\S_{\mathrm{qc}}(\H_{AB})$: their right hand sides tend to zero  as $\,\varepsilon\to0$. Moreover, for each $E>0$ and $\varepsilon\in(0,1]$ there exist  states $\rho$ and $\sigma$ in $\S_{\mathrm{qc}}(\H_{AB})$ such that}
$$
S(A\vert B)_{\rho}-S(A\vert B)_{\sigma}>\varepsilon F_H(E/\varepsilon),\quad \Tr H\rho_A,\Tr H\sigma_A\leq E\quad \textit{and} \quad \textstyle\frac{1}{2}\|\rho-\sigma\|_1\leq \varepsilon.
$$
\end{proposition}

\noindent\textbf{Note:} The  quantity $E_{H,\shs\varepsilon}(\rho)$
monotonously tends to $E\doteq\Tr H\rho_A\,$  as $\,\varepsilon\to 0$.
So, the first estimate in (\ref{qce-qc-2}) may be essentially  sharper than the second one for small $\,\varepsilon$.
\smallskip

The last claim of  Proposition \ref{qce-qc-1}
shows that the semicontinuity bound  (\ref{qce-qc-2}) is \emph{close to optimal}. It is  applicable
to any state  $\rho$ in $\S_{\mathrm{qc}}(\H_{AB})$ with finite $S(\rho_A)$, since for any such state there is a positive operator $H$ on $\H_A$ satisfying conditions (\ref{H-cond}) and (\ref{star}) such that $\Tr H\rho_A<+\infty$ \cite[Proposition 4]{EC}. A simple semicontinuty bound  for the QCE on $\S_{\mathrm{qc}}(\H_{AB})$ applicable
to any state  $\rho$ in $\S_{\mathrm{qc}}(\H_{AB})$ with finite $S(\rho_A)$ can be obtained by using Lemma \ref{b-lemma}A in Section 3. \smallskip

\begin{proposition}\label{CE-SCB-2}
\emph{Let $\rho$ be a state in $\,\S_{\mathrm{qc}}(\H_{AB})$ such that $\,S(A|B)_{\rho}<+\infty$. Then
\begin{equation}\label{CE-SCB-2+}
S(A\vert B)_{\rho}-S(A\vert B)_{\sigma}\leq \tilde{S}(A|B)_{\rho\wedge\varepsilon}+h^{\uparrow}(\varepsilon)
\end{equation}
for any state $\sigma$ in $\S_{\mathrm{qc}}(\H_{AB})$ such that $\,\frac{1}{2}\|\rho-\sigma\|_1\leq\varepsilon\in(0,1]$, where $\tilde{S}(A|B)$ is the homogeneous extension of the QCE to the cone generated by the set $\S_{\mathrm{qc}}(\H_{AB})$
of all q-c states  defined according to the rule (\ref{G-ext})), $h^{\uparrow}(\varepsilon)$ is the function defined in (\ref{h+}) and
$$
\rho\wedge\varepsilon\doteq\sum_k \left((p_k\rho_k)\wedge\varepsilon I_{A}\right)\otimes|k\rangle\langle k|
$$
\centerline{(here $(p_k\rho_k)\wedge\varepsilon I_{A}$ is the operator in $\T(\H_A)$ defined in (\ref{2-op})).}}\smallskip

\emph{The semicontinuity bound (\ref{CE-SCB-2+}) is faithful: its right hand side  tends to zero  as $\,\varepsilon\to0$.}
\end{proposition}\smallskip

\emph{Proof.} It suffices to repeat the arguments from the proof of the second claim of Proposition 5 in \cite{LLB} with the use
of Lemma \ref{b-lemma}A in Section 3 (instead of Lemma 2A in \cite{LLB}). $\Box$ \smallskip

By using expression (\ref{ce-rep}) Proposition \ref{CE-SCB-2} can be rewritten as follows\smallskip

\begin{corollary}\label{CE-LLB-c}  \emph{Let $\{p_k,\rho_k\}$ be an  ensemble of
states in $\S(\H)$ such that $\,\sum_kp_kS(\rho_k)<+\infty\,$ then
\begin{equation}\label{CE-LB-3++c}
\sum_kp_kS(\rho_k)-\sum_kq_kS(\sigma_k)\leq\sum_k\tilde{S}((p_k\rho_k)\wedge\varepsilon I_{\H})+h^{\uparrow}(\varepsilon)
\end{equation}
for any ensemble $\{q_k,\sigma_k\}$ of states in $\S(\H)$ such that  $\frac{1}{2}\sum_k\|p_k\rho_k-q_k\sigma_k\|_1\leq\varepsilon$,
where $\tilde{S}$  is the extended von Neumann  entropy defined in (\ref{S-ext}) and $\,h^{\uparrow}$ is the function defined in (\ref{h+}).}
\end{corollary} \smallskip

The semicontinuity  bound (\ref{CE-LB-3++c}) can be treated as a generalization of the  semicontinuity  bound (\ref{S-SCB-2+}) for the von Neumann entropy.
\smallskip

Lemma \ref{b-lemma}B in Section 3 implies the following  local lower bound for the QCE on the set $\S_{\mathrm{qc}}(\H_{AB})$ of all q-c states having the form (\ref{qc-states}).\smallskip

\begin{proposition}\label{CE-LLB}
\emph{Let $\rho$ be an arbitrary state in $\S_{\mathrm{qc}}(\H_{AB})$. Then
\begin{equation}\label{CE-LLB+}
S(A\vert B)_{\sigma}\geq \tilde{S}(A|B)_{\rho\ominus\varepsilon}-h^{\uparrow}(\varepsilon)
\end{equation}
for any  state $\sigma$ in $\S_{\mathrm{qc}}(\H_{AB})$ such that  $\,\frac{1}{2}\|\rho-\sigma\|_1\leq \varepsilon\in(0,1]$,
where $\tilde{S}(A|B)$ is the homogeneous extension of the QCE to the cone generated by the set $\S_{\mathrm{qc}}(\H_{AB})$
of all q-c states defined according to the rule (\ref{G-ext}), $h^{\uparrow}$ is the function defined in (\ref{h+}) and
$$
\rho\ominus\varepsilon\doteq\sum_k [p_k\rho_k-\varepsilon I_{A}]_+\otimes|k\rangle\langle k|
$$
\centerline{(here $[p_k\rho_k-\varepsilon I_{A}]_+$  is the
positive part of the Hermitian  operator $\,p_k\rho_k-\varepsilon I_{A}$).}}\smallskip

\emph{The lower bounds (\ref{CE-LLB+}) is faithful: its right hand sides tend to $\,S(A\vert B)_{\rho}\leq+\infty\,$ as $\,\varepsilon\to0$.}
\end{proposition}\smallskip

\emph{Proof.} It suffices to repeat the arguments from the proof of the first  claim of Proposition 5 in \cite{LLB} with the use
of Lemma \ref{b-lemma}B in Section 3 (instead of Lemma 2B in \cite{LLB}). $\Box$ \smallskip

\begin{remark}\label{CE-LLB-r}
The lower bound (\ref{CE-LLB+}) is universal: it is applicable to any q-c state $\rho$
(including q-c states at which QCE is equal to $+\infty$). Another advantage of the lower bound (\ref{CE-LLB+})
consists in the fact that its r.h.s.  is easily calculated, since the series in the r.h.s. of (\ref{CE-LLB+}) contains a finite number of nonzero summands and
all the operators  $[p_k\rho_k-\varepsilon I_{\H}]_+$ have finite rank.
\end{remark}\smallskip

By using expression (\ref{ce-rep}) Proposition \ref{CE-LLB} can be rewritten as follows\smallskip

\begin{corollary}\label{CE-LLB-ctt}  \emph{Let $\{p_k,\rho_k\}$ be an arbitrary ensemble of
states in $\S(\H)$. Then
\begin{equation}\label{S-LLB++}
\sum_kq_kS(\sigma_k)\geq \sum_k\tilde{S}([p_k\rho_k-\varepsilon I_{\H}]_+)-h^{\uparrow}(\varepsilon)
\end{equation}
for any ensemble $\{q_k,\sigma_k\}$ of states in $\S(\H)$ such that  $\frac{1}{2}\sum_k\|p_k\rho_k-q_k\sigma_k\|_1\leq\varepsilon$,
where $\tilde{S}$  is the extended von Neumann  entropy defined in (\ref{S-ext}) and $h^{\uparrow}$ is the function defined in (\ref{h+}).}
\end{corollary} \smallskip

The lower bound (\ref{S-LLB++}) can be treated as a generalization of the lower bound (\ref{S-LLB+}) for the von Neumann entropy.

\subsection{Entanglement of Formation}

In this section we show how the semicontinuity bounds for the Entanglement of Formation obtained in \cite{LCB} can be  improved
by using the results of Section 4.4.

The Entanglement of
Formation (EoF) is one of the basic entanglement measures in bipartite quantum systems \cite{Bennett,4H,P&V}. In a finite-dimensional bipartite system $AB$ the EoF is defined as the convex roof extension to the set $\S(\H_{AB})$ of the function $\omega\mapsto S(\omega_{A})$ on the set $\mathrm{ext}\shs\S(\H_{AB})$ of pure states in $\S(\H_{AB})$ , i.e.
\begin{equation*}
  E_{F}(\omega)=\inf_{\sum_kp_k\omega_k=\omega}\sum_k p_kS([\omega_k]_{A}),
\end{equation*}
where the infimum (which is always attained) is over all finite ensembles $\{p_k, \omega_k\}$ of pure states in $\S(\H_{AB})$ with the average state $\omega$ \cite{Bennett}.
In this case $E_F$ is a uniformly continuous function on $\S(\H_{AB})$  possessing basic properties of an entanglement measure \cite{Nielsen,4H,P&V}.

If both systems $A$ and $B$ are infinite-dimensional then  there are two versions $E_F^d$ and $E_F^c$ of the EoF defined, respectively, by means of discrete and continuous convex roof extensions
\begin{equation}\label{E_F-def-d}
E_F^d(\omega)=\!\inf_{\sum_k\!p_k\omega_k=\omega}\sum_kp_kS([\omega_k]_A),\quad
\end{equation}
\begin{equation}\label{E_F-def-c}
E_F^c(\omega)=\!\inf_{\int\omega'\mu(d\omega')=\omega}\int\! S(\omega'_A)\mu(d\omega'),
\end{equation}
where the  infimum in (\ref{E_F-def-d}) is over all countable ensembles $\{p_k, \omega_k\}$ of pure states in $\S(\H_{AB})$ with the average state $\omega$ and the  infimum in (\ref{E_F-def-c}) is over all Borel probability measures on the set of pure states in $\S(\H_{AB})$ with the barycenter $\omega$ (the last infimum is always attained). It follows from the definitions that $E_F^d(\omega)\geq E_F^c(\omega)$, while it is known that $E_F^d(\omega)=E_F^c(\omega)$ for any state $\omega$ in $\S(\H_{AB})$ such that $\min\{S(\omega_{A}),S(\omega_{B})\}<+\infty$ \cite[Section 4.4]{QC}. The conjecture of coincidence of $E_F^d$ and $E_F^c$ on $\S(\H_{AB})$ is an interesting open question.

By repeating  the arguments from the proof of Proposition 4 in \cite{LCB} with the use of Proposition \ref{qce-qc-1}  (instead of Proposition 3 in \cite{LCB})
one can obtain the semicontinuity bounds for the EoF presented in the following proposition (improved version of Proposition 4 in \cite{LCB}).\smallskip

\begin{proposition}\label{EF-SCB} \emph{Let $AB$ be an infinite-dimensional bipartite quantum system.}\smallskip

\noindent A) \emph{If $\rho$ is a state in $\S(\H_{AB})$ such that $\rank\rho_A$ is finite then
\begin{equation*}
   E^*_F(\rho)-E^*_F(\sigma)\leq \{\delta\ln(\rank\rho_A)+h(\delta)\}_{\delta}^{\uparrow}\leq\delta\ln(\rank\rho_A)+h^{\uparrow}(\delta),\quad E^*_F=E_F^d,E_F^c,
\end{equation*}
for any state $\sigma$ in $\S(\H_{AB})$ such that $\,\frac{1}{2}\|\rho-\sigma\|_1\leq \varepsilon\leq1$,  where $\delta=\sqrt{\varepsilon(2-\varepsilon)}$,  $\,\{...\}^{\uparrow}_{\delta}$ is defined in (\ref{men-def+}) and $\,h^{\uparrow}(\delta)$ is the function defined in (\ref{h+}).}
\smallskip

\noindent B) \emph{If $\rho$ is a state in $\S(\H_{AB})$ such that $E\doteq\Tr H\rho_A<+\infty$, where $H$ is a positive operator on $\H_A$  satisfying conditions (\ref{H-cond}) and (\ref{star}), then
\begin{equation}\label{EF-SCB-B}
   E^*_F(\rho)-E^*_F(\sigma)\leq \delta F_H(E/\delta)+h^{\uparrow}(\delta),\quad E^*_F=E_F^d,E_F^c,
\end{equation}
for any state $\sigma$ in $\S(\H_{AB})$ such that $\,\frac{1}{2}\|\rho-\sigma\|_1\leq \varepsilon\leq1$, where $\,\delta=\sqrt{\varepsilon(2-\varepsilon)}$ and
$h^{\uparrow}(\varepsilon)$ is the function defined in (\ref{h+}).}
\end{proposition}\smallskip

\begin{remark}\label{fidelity}
In both  semicontinuity bounds for the EoF presented in Proposition \ref{EF-SCB}
 the condition $\,\frac{1}{2}\|\rho-\sigma\|_1\leq \varepsilon$, where
$\varepsilon$ is such that $\,\delta=\sqrt{\varepsilon(2-\varepsilon)}$, can be replaced by the
condition  $\,F(\rho,\sigma)\geq 1-\delta^2$, where  $F(\rho,\sigma)\doteq\|\sqrt{\rho}\sqrt{\sigma}\|_1^2$
is the fidelity of $\rho$ and $\sigma$ \cite{H-SCI,Wilde}.  Such replacement  makes these semicontinuity bounds essentially sharper in some cases.
\end{remark}\medskip

An obvious application of the semicontinuity bounds in Proposition \ref{EF-SCB}  is the ability to obtain continuity
bounds for the Entanglement of Formation (the estimates of $|E^*_F(\rho)-E^*_F(\sigma)|$) with different constrains on the states $\rho$ and $\sigma$.
However, there are tasks where one-sided estimates of $E^*_F(\rho)-E^*_F(\sigma)$  are needed in the absence of any  constraints on the state $\sigma$.
An example of such a task is the proof of the converse part of the claim of Theorem 7 in \cite{Lami-new} where the initial (non-optimised) version of the  semicontinuity bound (\ref{EF-SCB-B}) presented in Proposition 4 in \cite{LCB} is used essentially. Theorem 7 in \cite{Lami-new} states the coincidence of the Entanglement Cost with the regularized  Entanglement of Formation for any state $\omega$ of an infinite-dimensional bipartite quantum system $AB$ such that $\,\min\{S(\omega_{A}),S(\omega_{B})\}<+\infty$.

\section{Other applications}

\subsection{Characteristics of discrete random variables}

Using the advanced version of the Alicki-Fannes-Winter technique in the quasi-classical settings considered in Section 3
one can improve the continuity, semicontinuity bounds and local lower bounds for characteristics  of discrete random variables
presented in \cite[Section 4.4]{LCB},\cite[Section 5]{LLB}. Indeed, by applying Corollaries 1 and 2 in Section 3.2 (instead of Corollaries 1 and 2 in \cite{LCB}) one can prove the
advanced versions of all the continuity and semicontinuity bounds under constrains of different types presented in \cite[Section 4.4]{LCB} in which the term $g(\varepsilon)$ in the right hand sides is replaced by
the term $h^{\uparrow}(\varepsilon)$ defined in (\ref{h+}). The same replacement can be done in all the local lower bounds  presented in \cite[Section 5]{LLB}. This can be shown by applying the advanced version of Lemma 3 in \cite{LLB} obtained  with the use of Lemma 3 in Section 3.1 of this article.

\subsection{Characteristics of classical states of quantum oscillators}

Using the result from Section 3
one can improve the continuity and semicontinuity bounds for characteristics of classical states of a multi-mode quantum oscillator
presented in \cite[Section 4.5]{LCB}. Indeed, by applying Theorem \ref{main-2} in Section 3.2 (instead of Theorem 2 in \cite{LCB}) one can prove the
advanced version of Proposition 7 in \cite{LCB} in which the term $g(\varepsilon)$ in the in the right hand sides of both inequalities is replaced by
the term $h^{\uparrow}(\varepsilon)$ defined in (\ref{h+}). Accordingly, the same replacement can be done in all the inequalities in Example 3 in \cite{LCB}.
\bigskip\bigskip

I am grateful to A.S.Holevo and to the participants of his seminar  "Quantum Probability, Statistics, Information" (the Steklov  Mathematical Institute) for useful discussion.
Special thanks to the authors of the recent article \cite{D++}, who suggested to me a trick that was significantly used in Section 3 of this article.

\end{document}